\documentclass[a4paper,UKenglish,cleveref, autoref, thm-restate]{lipics-v2021}

\pdfoutput=1 
\hideLIPIcs  
\nolinenumbers

\newcommand{\mytitle}{Higher-Order Specifications\\ for Deductive
  Synthesis of Programs with Pointers (Extended Version)}

\newcommand{\runningtitle}{Higher-Order Specifications for Deductive
  Synthesis of Programs with Pointers}

\title{\mytitle}
\titlerunning{\runningtitle}

\newcommand{\equalcontribution}{Contributed equally to this work.}

\author{David Young\text{*}}{University of Kansas, USA}{d063y800@ku.edu}{https://orcid.org/0009-0006-1193-330X}{}
\author{Ziyi Yang\text{*}}{National University of Singapore, Singapore}{yangziyi@u.nus.edu}{https://orcid.org/0000-0002-8015-7846}{}
\author{Ilya Sergey}{National University of Singapore, Singapore}{ilya@nus.edu.sg}{https://orcid.org/0000-0003-4250-5392}{}
\author{Alex Potanin}{Australian National University, Australia}{alex.potanin@anu.edu.au}{https://orcid.org/0000-0002-4242-2725}{}

\authorrunning{Young et al.}

\Copyright{David Young, Ziyi Yang, Ilya Sergey, Alex Potanin} 

\keywords{Program Synthesis, Separation Logic, Functional Programming} 

\category{} 

\relatedversion{} 

\EventEditors{Jonathan Aldrich and Guido Salvaneschi}
\EventNoEds{2}
\EventLongTitle{38th European Conference on Object-Oriented Programming (ECOOP 2024)}
\EventShortTitle{ECOOP 2024}
\EventAcronym{ECOOP}
\EventYear{2024}
\EventDate{September 16--20, 2024}
\EventLocation{Vienna, Austria}
\EventLogo{}
\SeriesVolume{313}
\ArticleNo{34}

\usepackage{amsmath}
\usepackage{graphicx}
\usepackage{mathpartir}
\usepackage{xspace}
\usepackage{bbding}
\usepackage{pifont}
\usepackage{syntax}
\usepackage{float}
\usepackage{proof}
\usepackage{stmaryrd}
\usepackage{calc}
\usepackage[most]{tcolorbox}
\usepackage{tikz-cd}
\usepackage{mathtools}
\usepackage{thm-restate}
\usepackage{cite}
\usepackage{booktabs}

\hypersetup{colorlinks,
  linkcolor=ACMDarkBlue,
  citecolor=ACMPurple,
  urlcolor=ACMDarkBlue,
  filecolor=ACMDarkBlue}

\floatname{algorithm}{Algorithm}

\makeatletter 
\def\arcr{\@arraycr}
\makeatother

\definecolor{ACMBlue}{cmyk}{1,0.1,0,0.1}
\definecolor{ACMYellow}{cmyk}{0,0.16,1,0}
\definecolor{ACMOrange}{cmyk}{0,0.42,1,0.01}
\definecolor{ACMRed}{cmyk}{0,0.90,0.86,0}
\definecolor{ACMLightBlue}{cmyk}{0.49,0.01,0,0}
\definecolor{ACMGreen}{cmyk}{0.20,0,1,0.19}
\definecolor{ACMPurple}{cmyk}{0.55,1,0,0.15}
\definecolor{ACMDarkBlue}{cmyk}{1,0.58,0,0.21}
\definecolor{PaleGreen}{RGB}{196, 255, 231}
\definecolor{PaleOrange}{RGB}{255, 213, 169}
\definecolor{intnull}{RGB}{213,229,255}

\definecolor{shadecolor}{gray}{1.00}
\definecolor{ddarkgray}{gray}{0.5}
\definecolor{darkgray}{gray}{0.30}
\definecolor{light-gray}{gray}{0.91}

\newcommand{\ie}{\emph{i.e.}\xspace}

\newcommand{\eg}{\emph{e.g.}\xspace}

\newcommand{\etal}{\emph{et~al.}\xspace}

\newcommand{\cf}{\textit{cf.}\xspace}

\newcommand{\tname}[1]{\textsf{#1}\xspace}
\newcommand{\tool}{\tname{Pika}}
\newcommand{\suslik}{\tname{SuSLik}}
\newcommand{\fnLang} {tool\xspace}

\newcommand{\theFnLang} {the tool\xspace}

\newcommand{\fsImpl}{\textbf{\fnLang implementation}}
\newcommand{\genSuSLik}{\textbf{Generated SuSLik}}
\newcommand{\refSuSLik}{\textbf{Reference SuSLik}}

\newcommand{\generate}{\texttt{\%generate}}

\newcommand{\stage} [1] {\textit{\textbf{#1}}}

\newcommand{\Ra} {\Rightarrow}

\newcommand{\step} {\ensuremath{\longmapsto}}
\newcommand{\tstep} {\ensuremath{\Downarrow}}
\newcommand{\defstep} [1] {\ensuremath{\xmapsto{\textnormal{fn-def}}_{#1}}}

\newcommand{\Expr} {\textnormal{Expr}}

\newcommand{\Var} {\textnormal{Var}}
\newcommand{\Loc} {\textnormal{Loc}}
\newcommand{\Val} {\textnormal{Val}}
\newcommand{\FsVal} {\textnormal{FsVal}}
\newcommand{\labinfer} [3] [] {\infer[{\textsc{#1}}]{#2}{#3}}
\newcommand{\fresh} {\mathbin{\#}}

\newcommand{\freshLoc} [1] {#1\;\textnormal{fresh}}
\newcommand{\freshVar} [1] {#1\;\textnormal{fresh}}

\newcommand{\FsStore} {\textnormal{FsStore}}
\newcommand{\VarSet} {V}

\newcommand{\concrete} [1] {\textnormal{\;concrete$_{#1}$}}

\newcommand{\Int}{\texttt{Int}}
\newcommand{\Bool}{\texttt{Bool}}

\newcommand{\sep}{*}

\newcommand{\Pure}{\textnormal{Pure}}
\newcommand{\Spatial}{\textnormal{Spatial}}

\newcommand{\monic}{\ensuremath{\rightarrowtail}}

\newcommand{\layout}{\texttt{layout}}

\newcommand{\lowerS} [2] {\texttt{lower}_{#1}(#2)}
\newcommand{\instantiateS} [2] {\texttt{inst}_{#1}(#2)}

\newcommand{\astRow} [3] {#1 & #2 & #3 & \the\numexpr #3 / #2}

\newcommand{\emp} {\texttt{emp}}
\newcommand{\pointsto} {\ensuremath{\mathbin{\texttt{:->}}}}

\newcommand{\sem} [1] {\llbracket#1\rrbracket}
\newcommand{\Tsem} [2] {\mathcal{T}\sem{#1}_{#2}}

\newcommand{\cond}{\textnormal{cond}}

\newcommand{\Start} [1] {{\color{cyan} #1}}
\newcommand{\End} [1] {{\color{orange} #1}}

\newcommand{\StartSigma} {\Start{\sigma}}
\newcommand{\EndSigma} {\End{\sigma'}}

\newcommand{\EndPair} {(\EndSigma, \End{h'})}

\newcommand{\emptyheap} {\varnothing}

\newcommand{\dom} {dom}

\newcommand{\RQ} [1] {\textbf{RQ#1}}

\definecolor{pblue}{rgb}{0.13,0.13,1}
\definecolor{pgreen}{rgb}{0,0.5,0}
\definecolor{pred}{rgb}{0.7,0,0}
\definecolor{pgrey}{rgb}{0.46,0.45,0.48}

\lstdefinelanguage{SynLang}{
  keywords={new, let, if, else, null, return, while},
  ndkeywords={bool, int, void, loc},
  mathescape=true,
  showspaces=false,
  showtabs=false,
  breaklines=true,
  showstringspaces=false,
  breakatwhitespace=true,
  lineskip=-0.9pt,
  morecomment=[l]{//}, 
  morecomment=[s]{/*}{*/}, 
  basewidth={0.54em, 0.4em},%
  basicstyle=\footnotesize\ttfamily,
  keywordstyle={\color{pred}\ttfamily\bfseries},
  ndkeywordstyle={\color{pblue}\ttfamily\bfseries},
  commentstyle={\color{ccomment}\itshape},
  numbers=none,
}

\newcommand{\pts}{\mapsto}
\newcommand{\code}[1]{\lstinline[language=SynLang,basicstyle=\small\ttfamily,mathescape=true]{#1}}

\newcommand{\rulename}[1]{\textsc{#1}}
\newcommand{\func}{\textit{func}\xspace}
\newcommand{\writer}{\rulename{Write}\xspace}

\newcommand{\funcwrite}{\rulename{Funcwrite}\xspace}

\newcommand{\tempfuncalloc}{\rulename{Tempfuncalloc}\xspace}

\newcommand{\tempfuncfree}{\rulename{Tempfuncfree}\xspace}

\newcommand{\mcode}[1]{{\ensuremath{\tt #1}}}
\newcommand{\vars}[1]{\mathsf{Vars}\left({#1}\right)}

\newcommand{\set}[1]{\left\{{#1}\right\}}
\newcommand{\bl}[1]{{\color{pblue}{#1}}}
\newcommand{\asn}[1]{{\bl{\set{#1}}}}

\newcommand{\trans}[3]{\left.\bl{#1} \!\leadsto\! \bl{#2} \right| #3}

\newcommand{\prog}{c}
\newcommand{\osep}{\ast}
\newcommand{\env}{\Gamma}
\newcommand{\deref}[1]{*{#1}}

\newcommand{\iann}[1]{\mcode{ {#1}}} 
\newcommand{\writeontop}[2]{\mathrel{\stackrel{\makebox[0pt]{\mbox{\tiny{#2}}}}{#1}}}

\newcommand{\ispointsto}[3]{#1{\writeontop{\ \!\pts\ \!}{\iann{#3}}}{#2}} 

\newcommand{\ruleSAdd} {
  \labinfer[S-Add]{(e_1 + e_2, \VarSet_0) \tstep (v == v_1 + v_2 \land p_1 \land p_2, s_1 \sep s_2, \VarSet_2 \cup \{v\}, v) }
    {(e_1, \VarSet_0) \tstep (p_1, s_2, \VarSet_1, v_1) &
     (e_2, \VarSet_1) \tstep (p_2, s_2, \VarSet_2, v_2) &
     \freshVar{v}
    }}

\begin{document}

\maketitle
\renewcommand{\thefootnote}{\fnsymbol{footnote}}
\footnotetext[1]{\equalcontribution}
\renewcommand{\thefootnote}{\arabic{footnote}}

\begin{abstract}
  Synthetic Separation Logic (SSL) is a formalism that powers \suslik,
the state-of-the-art approach for the deductive synthesis of
provably-correct programs in C-like languages that manipulate
heap-based linked data structures.
Despite its expressivity, SSL suffers from two shortcomings that
hinder its utility.
First, its main specification component, inductive predicates, only
admits \emph{first-order} definitions of data structure shapes, which
leads to the proliferation of ``boiler-plate'' predicates for
specifying common patterns.
Second, SSL requires \emph{concrete} definitions of data structures to
synthesise programs that manipulate them, which results in the need
to change a specification for a synthesis task every time changes are
introduced into the layout of the involved structures.


We propose to significantly lift the level of abstraction used in
writing Separation Logic specifications for synthesis---both
simplifying the approach and making the specifications more usable and
easy to read and follow. 
We avoid the need to repetitively re-state low-level representation
details throughout the specifications---allowing the reuse of
different implementations of the same data structure by abstracting
away the details of a specific layout used in memory. Our novel
\textit{high-level front-end language} called~\tool significantly improves the expressiveness of \suslik.

We implemented a layout-agnostic synthesiser from \tool to \suslik
enabling push-button synthesis of C programs with in-place memory
updates, along with the accompanying full proofs that they meet
Separation Logic-style specifications, from high-level specifications
that resemble ordinary functional programs.
Our experiments show that our tool can produce C code that is
comparable in its performance characteristics and is sometimes faster
than Haskell.

\end{abstract}


\section{Introduction}

Recent advances in program synthesis have allowed programmers to
concentrate on stating precise specifications---leaving the job of
generating provably correct and efficient imperative code to the
synthesiser, such
as~\suslik~\cite{polikarpova:2019:suslik,itzhaky:2021:cyclic-synth,WatanabeGPPS21}.
Such specifications are usually expressed using (Synthetic) Separation
Logic~\cite{ohearn:2001:seplogic,reynolds:2002:seplogic} that while
hugely successful in verifying properties of pointer-manipulating
programs remains out of reach to many mainstream developers.
As the programs grow in complexity, such SSL specifications can become
exceedingly verbose and complex---making the job of specification
writer especially error-prone and defeating the purpose of a
\textit{usable proof automation toolchain}.

The power of Separation Logic (SL) specifications for the tasks of both
verification and synthesis, is its mechanism of \emph{inductive
  predicates} that concisely capture the shape of possibly recursive
pointer-based tree-like data structures, determining both induction
schemes for verification and the shape of recursion for the synthesis
tasks of the programs that manipulate such data
structures~\cite{reynolds:2002:seplogic}.
The surprising ability of SL specifications to capture precisely the
logic of a desired program to be synthesised in the logical assertions
comes at the price of the involved inductive predicates being (a)
\emph{first-order} and (b) somewhat \emph{low-level}, with both these
aspects posing limitations to the usability of SL-based program
synthesis.

The first-order nature of the predicates means, for instance, that
synthesis tasks that involve several data structures with very similar
heap layouts would require to use \emph{different} predicate
definitions.
As a specific example, consider a task synthesising two functions,
\code{f} and \code{g}. Both \code{f} and \code{g} take as an argument
a pointer to a linked list of integers; \code{f} increments all
elements of the list by one, while \code{g} multiplies all its
elements by two.
To specify these two tasks, the state-of-the-art tools for program
synthesis based on SL specifications require the user to provide, in
addition to the pre-/postconditions, \emph{three} different inductive
predicates: one for an arbitrary list, another for a list that carries
a known payload, with each element incremented by one, and the final
capturing the multiplication of each element by two.

The second aspect, \ie, the low-level nature of the SL inductive
predicates used for synthesis, shows up when we try to \emph{rewrite}
an already specified synthesis task for a data structure with a slightly
different layout.
As an example, imagine defining the task of concatenating two lists.
It is natural to expect that the specification will look very similar
for both singly and doubly linked lists. Yet, since those are two
different structures, with two different layouts, the user would
require to supply two different task specifications.

A seasoned programmer would immediately notice that both issues
(a)~and (b)~very much resemble the struggle that one faces
when programming in a language that does not provide certain
abstractions, which are nowadays mostly taken for granted:
\emph{higher-order functions} and \emph{abstract data types}.
Our first example could be streamlined should the SL-based
specification language offer a way to define a function similar to
\code{List.map}, available in all popular functional programming
languages, so it could be used to concisely express the two scenarios
of manipulating with the payload of a list's elements.
The second example would benefit from the ability to specify
concatenation of \emph{abstract} lists, while separately handling
manipulations with single- or doubly-linked lists in terms of their
memory \emph{layouts}.


\subsubsection*{Key Ideas}
\label{sec:key-ideas}

The two challenges faced by the SSL specification language for the
synthesis of heap-manipulating programs---the need for higher-order
functions and abstract data types---have provided the primary
motivation for this work.

As a solution, we developed \tool: a high-level front-end language and
specification translation framework built on to of the
state-of-the-art SL-based program synthesiser \suslik.\footnote{ Both
  susliks and pikas are Central Asian mammals. Pikas might look
  similar to susliks, but are more nimble and have longer life
  expectancy. }
\suslik is based on a variant of separation logic called
\textit{synthetic separation logic} or
SSL~\cite{polikarpova:2019:suslik}.
\tool has a syntax similar to popular functional programming languages
and features specification-level higher-order functions on Algebraic
Data Types (ADTs). In addition to being more succinct in comparison to
Synthetic Separation Logic (SSL), the specification formalism of
\suslik, \tool also addresses its reusability issues outlined above.
First, the use of specification-level higher-order functions allows
the user to abstract over the specific properties of the payloads of
the heap-based data structures, thus, generalising existing inductive
predicates for so they could be employed in a wider range of synthesis
tasks.
Second, by manipulating ADTs, the synthesis specifications do not need
to deal with the specific memory layouts of the data structures. This
separates the specification of the tasks that operate on those data
structures from the low-level details of their memory representations.

\begin{figure}[t]
  {\footnotesize
\[\begin{tikzcd}
	{\fbox{\textnormal{\tool code}}} & {\fbox{\suslik specification}} &
  {\fbox{\suslik synthesis tree}} & {\fbox{\textnormal{C code + SL proof}}}
	\arrow[from=1-1, to=1-2]
	\arrow[from=1-2, to=1-3]
	\arrow[from=1-3, to=1-4]
\end{tikzcd}\]
  \caption{\tool translation pipeline}
  \label{fig:pipeline}
  }
\end{figure}

The \tool pipeline is depicted in \autoref{fig:pipeline}.
As our goal was to extend the expressivity of \suslik by giving it a
high-level specification language while retaining its meta-theoretical
guarantees (\ie, the certifiable correctness of the programs it
synthesises), we had to overcome several technical obstacles when
designing \tool and implementing it on top of \suslik.
First, we had to formally define the operational semantics of \tool
and its translation to SSL-based specifications of \suslik, defining
the corresponding soundness result, relating the behaviour of programs
eventually synthesised to that of their high-level counterparts
(\autoref{sec:formal}).
Second, to support the higher-order specifications of \tool, we
introduced conservative extensions to \suslik's specification language
as well as to its deductive synthesis rules, to make it capable of
handling pre-/postconditions with first-class functions operating on
the payload of heap-stored data types (\autoref{sec:extensions}).

\subsubsection*{\tool as a Programming Language}
\label{sec:benefits}

Given the close similarity of \tool, our new specification language,
to general-purpose functional programming languages, such as Haskell,
it is natural to wonder whether it's possible to leverage its underlying
synthesis pipeline as a way to produce efficient imperative programs
in a language such as C from equivalent high-level functional
programs.
In other words, if we decide to use a combination \tool + \suslik as a
\emph{compiler}, would it be a viable replacement to many-decades old
tools such as Glasgow Haskell Compiler (GHC)~\cite{HallHPJW92}, for
producing efficient runnable code?
To answer this question, we have conducted evaluation on several list-
and tree-manipulating benchmarks, comparing the performance of
verified C code, emitted by \suslik from \tool specifications, to that
of executables produced for equivalent tasks by GHC.

Our preliminary results are encouraging: thanks to avoiding
unnecessary allocation and using destructive heap updates whenever
possible, the C code synthesised by \tool + \suslik outperforms the
GHC output (with compiler optimisations flags turned) in the majority
of list-manipulating benchmarks we've tried; our synthesis tool also
produces strictly more performant C code for tree-manipulating
benchmarks when compared to the corresponding Haskell programs,
compiled by GHC without or with optimisations
(\autoref{sec:evaluation}). In particular, we specifically observe
this when we compare C code synthesised by \tool + \suslik with no C
compiler optimisation flags to GHC compiled code with no GHC compiler
flags.

\subsubsection*{Contributions}
\label{sec:contribs}

In this work, we make the following contributions:

\begin{itemize}
\item We address the expressivity limitations of SSL, the
  specification formalism of the state-of-the-art deductive synthesis
  tool \suslik by developing \tool---a high-level specification
  language with higher-order functions and abstract data types.

\item We formally define the operational semantics of \tool and prove
  the soundness of translation from \tool to pre-/postconditions in
  SSL.

\item We develop an extension to \suslik's specification and synthesis
  mechanism that enables translation from \tool specifications
  featuring first-class functions.

\item We observe that the synthesis tool resulting from the
  combination \tool + \suslik enables certified memory-layout-agnostic
  compilation from a functional specification to C code.

  \item We report on the evaluation of the \tool regarding its
    expressiveness and performance. In particular, we show that the C
    code it produces frequently outperforms equivalent Haskell
    programs compiled by GHC.
\end{itemize}

\section{Overview}

\subsection{Background}

The \tool language is translated into
\suslik~\cite{polikarpova:2019:suslik}, which is a program synthesis
tool that uses Synthetic Separation Logic (SSL)-a variant of
Hoare-style~\cite{hoare:1969:axiomatic} Separation Logic
(SL)~\cite{ohearn:2001:seplogic}.


A synthesis task specification in \suslik is given as a function
signature together with a pair of pre- and post-conditions, which are
both SL assertions~\cite{polikarpova:2019:suslik}.
The synthesiser generates code that satisfies the given specification,
along with the SL \emph{proof} of its correctness by searching in a
space of proofs that can be derived by using the rules of the
underlying logic~\cite{WatanabeGPPS21}.
A distinguishing feature of \textit{Synthetic} Separation Logic is the
format of its assertions. An SSL assertion consists of two parts: a
\textit{pure part} and a \textit{spatial part}. The pure part is a
Boolean expression constraining the variables of the specification
using a few basic relations, like equality and the less-than relation
for integers. The spatial part is a \textit{symbolic heap}, which
consists of a list of \emph{heaplets} separated by the $\sep$ symbol.

Each heaplet takes on one of the following
forms~\cite{polikarpova:2019:suslik}:

\begin{itemize}
  \item \texttt{emp}: This represents the empty heap. It is also the left and right identity for $\sep$.
  \item $\ell \mapsto a$: This asserts that memory location $\ell$ points to the value $a$. It also asserts that the location $\ell$ is accessible.
  \item $[\ell, \iota]$: At memory location $\ell$, there is a block of size $\iota$.
  \item $p(\overline{\phi})$: This is an application of the
    \textbf{\textit{inductive predicate}} $p$ to the arguments
    $\overline{\phi}$. An inductive predicate has a collection of
    branches, guarded by Boolean expressions (conditions) on the
    parameters. The body of each branch is an SSL assertion. The
    assertion associated with the first branch condition that is
    satisfied is used in place of the application. Note that inductive
    predicates can be (and often are) recursively defined.
\end{itemize}

\noindent
The general form of an SSL assertion is
$(p; h_1 \sep h_2 \sep \cdots \sep h_n)$, where $p$ is the pure part
and $h_1, h_2, \cdots, h_n$ are the heaplets which are the conjuncts
that make up a separated conjunction of the spatial part. $\sep$ is \textit{separating conjunction}: $h_1 \sep h_2$ means that
the heaplets $h_1$ and $h_2$ apply to disjoint parts of the heap. A syntax
definition for SSL is given in \autoref{fig:SSL-syntax}, which is
adapted from \textit{Cyclic Program Synthesis} by Itzhaky
\etal~\cite{itzhaky:2021:cyclic-synth}. We will also use the symbol \verb|**| for separating conjunction and the symbol \verb|:->| for $\mapsto$.

\begin{figure}[t]
\setlength{\abovecaptionskip}{5pt}
{\small{
\[
  \begin{array}[t]{ l l }
    \text{Variable} &x, y\;\; \text{Alpha-numeric identifiers $\in$ \textnormal{Var}}\\
    \text{Size, offset} &n, \iota\;\; \text{Non-negative integers}\\
    \text{Expression} &e ::= 0\; \mid \texttt{true} \mid x \mid e = e \mid e \land e \mid \neg e \mid d\\
    \text{$\mathcal{T}$-expr.} &d ::= n \mid x \mid d + d \mid n \cdot d \mid \{\} \mid {d} \mid \cdots\\
    \text{Command} &c ::= \begin{aligned}[t] &\texttt{let x = *(x + }\iota\texttt{)} \mid \texttt{*(x + }\iota\texttt{) = e} \mid
      \\
      &\texttt{let x = malloc(n)} \mid \texttt{free(x)} \mid \texttt{error}\newline
      \mid \texttt{f(}\overline{e_i}\texttt{)}
      \end{aligned}\\
    \text{Program} &\Pi ::= \overline{f(\overline{x_i})\; \{\; c\; \}\; ;}\; c\\
    \text{Logical variable} &\nu, \omega\\
    \text{Cardinality variable} &\alpha\\
    \text{$\mathcal{T}$-term} &\kappa ::= \nu \mid e \mid \cdots\\
    \text{Pure logic term} &\phi, \psi, \chi ::= \kappa \mid \phi = \phi \mid \phi \land \phi \mid \neg \phi\\
    \text{Symbolic heap} &P, Q, R ::= \texttt{emp} \mid \mbox{$\langle e, \iota \rangle \mapsto e \mid [e, \iota]$} \mid p^{\alpha}(\overline{\phi_i})
      \mid \mbox{$P * Q$}\\
    \text{Heap predicate} &\mathcal{D} ::= p^{\alpha}(\overline{x_i}) : \overline{e_j \Ra \exists \overline{y}.\{\chi_j;R_j\}}\\
    \text{Assertion} &\mathcal{P},\mathcal{Q} ::= \{\phi; P\}\\
    \text{Environment} &\Gamma ::= \forall\overline{x_i}.\exists\overline{y_j}.\\
    \text{Context} &\Sigma ::= \overline{\mathcal{D}}\\
    \text{Synthesis goal} &\mathcal{G} ::= P \leadsto Q
  \end{array}
\]
}}
\caption{Syntax of Synthetic Separation Logic}
  \label{fig:SSL-syntax}
\end{figure}


As the specification language of \suslik, SSL serves as the
compilation target for the \tool language. From there, executable
programs are generated through \suslik's program synthesis. Consider a
program that takes an integer $x$ and a result in location $r$ and
stores $x+1$ at location $r$. This can be written as the \suslik
specification:

\begin{lstlisting}
void add1Proc(int x, loc r)
  { r :-> 0 }
  { y == x + 1 ; r :-> y }
{ ?? }
\end{lstlisting}

\noindent
This example can be written as follows in our \fnLang:

\begin{lstlisting}
add1Proc : Int -> Int;
add1Proc x := x + 1;
\end{lstlisting}

\noindent
In contrast with \suslik spec, the \tool one requires no direct
manipulation with pointers.




\subsection{The \tool Language}
\label{sec:language}


While SSL provides a specification language that allows tools like
\suslik to synthesise code, it is only able to express specifications
as pointers. This is useful for some applications, such as embedded
systems, but it does not provide any high-level abstractions. As a
result, every part of a specification is tailored to a specific memory
representation of each data structure involved.
To address this shortcoming, we introduce a language with algebraic
data types that gets translated into SSL specifications. Additionally,
we introduce a language construct that allows the programmer to
specify a \emph{memory representation} of an algebraic data type
called a \textit{layout}. 
The distinction between algebraic data types and layouts provides the
separation of concerns between the low-level representation of a data
structure and code that manipulates it at a high level.

Syntactically, \tool resembles a functional programming language
of the Miranda~\cite{turner:1986:miranda} and
Haskell~\cite{hudak:2007:haskell} lineage. It supports algebraic data
types, pattern matching at top-level function definitions (though not
inside expressions) and Boolean guards. The primary difference arises
due to the existence of layouts and the fact that the language is
compiled to an SSL specification rather than executable code. Beyond
algebraic data types and layouts, Pika has a built-in type for integers
as well as Booleans.

Functions in \tool are only defined by their operations on algebraic
data types. Thus, all function definitions are ``layout-polymorphic''
over the particular choices of layouts for their arguments and result.
Giving a layout polymorphic function, a particular choice of layouts
is called ``instantiation''. Specifying the layout of a non-function
value is called ``lowering.''



The code generator is instructed to generate a \suslik specification
for a certain function at a certain instantiation by using a
\generate{} directive. For example, if there is a function definition
with the type signature \verb|mapAdd1 : List -> List|, a line reading
\verb|%generate mapAdd1 [Sll] Sll| would instruct the
\tool compiler to generate the \suslik inductive predicate
corresponding to \verb|mapAdd1| instantiated to the \verb|Sll| layout
for both its argument and its result.
An example of an ADT definition and a corresponding layout definition
is given in \autoref{fig:List-def}. There is one unusual part of the
syntax in particular that requires further explanation: layout type
signatures. A layout definition consists of a \textit{layout type
  signature} and a pattern match (much like a function definition),
with lists of SSL heaplets on the right-hand sides. A layout type
signature has a special form $A : \alpha \monic \layout[x]$. This says
that the layout $A$ is for the algebraic data type $\alpha$ and the
SSL output variable $x$ denoting the ``head'' pointer of the
respective structure. 


\begin{figure}[t]
\begin{lstlisting}
data List := Nil | Cons Int List;

Sll : List >-> layout[x];
Sll (Nil) := emp;
Sll (Cons head tail) := x :-> head, (x+1) :-> tail, Sll tail;
\end{lstlisting}
  \caption{\lstinline{List} algebraic data type together with its singly-linked list layout \lstinline{Sll}}
  \label{fig:List-def}
\end{figure}

\subsection{\tool by Example}

We demonstrate the characteristic usages of \tool by a series of
examples.
In these examples, we will often make use of the \verb|List| algebraic
data type and its \verb|Sll| layout from~\autoref{fig:List-def}. A
simple example of \tool code that illustrates algebraic data types and
layouts is a function which creates a singleton list out of the given
integer argument:

\begin{lstlisting}
%generate singleton [Int] Sll

singleton : Int -> List;
singleton x := Cons x (Nil);
\end{lstlisting}

\noindent
This gets compiled to the following \suslik specification (modulo
auto-generated names):

\begin{lstlisting}
predicate singleton(int p, loc r) {
| true => { r :-> p ** (r+1) :-> 0 ** [r,2] }
}
\end{lstlisting}

\noindent
A slightly more complicated example comes from trying to write a
functional-style \verb|map| function directly in \suslik. Consider a
function which adds 1 to each integer in a list of integers.
Considering the list implementation to be a singly-linked list with a
fixed layout, one way to express this in \suslik is shown in \autoref{fig:sllx}. In the
\verb|mapAdd1| predicate, the input list is given as the \verb|x| parameter and the \verb|r|
parameter points to the output list. In the non-null case, the head of \verb|r| is required
to be the successor of the head of \verb|x|. The predicate is then applied recursively to the tails.

\begin{figure}[t]
\begin{lstlisting}
predicate sll(loc x) {
| x == 0 => { emp }
| not (x == 0) => { [x, 2] ** x :-> v ** (x+1) :-> nxt ** sll(nxt) }
}

predicate mapAdd1(loc x, loc r) {
| x == 0 => { emp }
| not (x == 0) => { [x, 2] ** x :-> v ** (x+1) :-> xNxt ** [r, 2]
    ** r :-> (v+1) ** (r+1) :-> rNxt ** mapAdd1(xNxt, rNxt) } }

void mapAdd1_fn(loc x, loc y)
  { sll(x) ** y :-> 0 }
  { y :-> r ** mapAdd1(x, r) }
{ ?? }
\end{lstlisting}
\caption{Specifying a function that adds one to each element of a
  singly-linked list in \suslik.}
\label{fig:sllx}
\end{figure}

Note that inductive predicates are used for \textit{two} different purposes: the \verb|sll| inductive
predicate describes a singly-linked list data structure, while the \verb|mapAdd1| inductive predicate 
describes how the input list relates to the output list. Both are used in the specification of
\verb|mapAdd1_fn|: \verb|sll| in the precondition and \verb|mapAdd1| in the postcondition.

Using the \verb|mapAdd1| inductive predicate gives us two advantages over attempting to
put the SSL propositions directly into the postcondition of \verb|mapAdd1_fn|:
\begin{enumerate}
  \item We are able to express a conditional on the shape of the list. This is much like pattern
    matching in a language with algebraic data types, but we are examining the pointer involved directly.
  \item We are able to express \textit{recursion} part of the postcondition: the \verb|mapAdd1| inductive
    predicate refers to itself in the \verb|not (x == 0)| branch.
\end{enumerate}
These two features are both reminiscent of features common in
functional programming languages: pattern matching and recursion.
However, there are still some significant differences:

\begin{itemize}
  \item In traditional pattern matching, the underlying memory representation of the data structure is not exposed.
  \item Compared to a functional programming language, the meaning of the specification is more obscured. It is
    necessary to think about the structure of the linked data structure to determine what the specification is saying. This
    is related to the first point: The memory representation is front-and-center.
  \item In many functional languages, mutation is either restricted or generally discouraged. In \suslik, mutation is commonplace.
\end{itemize}

Suppose we want to write the functional program that corresponds to
this specification. One way to do this in a Haskell-like language is
by using the \verb|List| type from \autoref{fig:List-def}.

\begin{lstlisting}
mapAdd1_fn : List -> List;
mapAdd1_fn (Nil) := 0;
mapAdd1_fn (Cons head tail) := Cons (head + 1) (mapAdd1_fn tail);
\end{lstlisting}

\noindent
The only missing information is the memory representation of the \verb|List| data structure. We do not want the
\verb|mapAdd1_fn| implementation to deal with this directly, however. We want to separate the more abstract notions
of pattern matching and constructors from the concrete memory layout that the data structure has.

To accomplish this, we now extend the code with the definition of
\verb|Sll| from \autoref{fig:List-def}. \verb|Sll| is a
\textit{layout} for the algebraic data type \verb|List|.
Now we have all of the information of the original specification but
rearranged so that the low-level memory layout is separated from the
rest of the code. This separation brings us to an important
observation about the language, manifested throughout these examples:
none of the function definitions need to \emph{directly} perform any
pointer manipulations. This is relegated entirely to the reusable
layout definitions for the ADTs. The examples are written entirely as
recursive functions that pattern match on, and construct, ADTs.

All that is left is to connect these two parts: the layouts and the
function definitions. We instruct a \suslik specification generator to
generate a \suslik specification from the \verb|mapAdd1_fn| function
using the \verb|Sll| layout:
\begin{lstlisting}
%generate mapAdd1_fn [Sll] Sll
\end{lstlisting}
The \verb|[Sll]| part of the directive tells the generator which
layouts are used for the arguments. In this case, the function only
has one argument and the \verb|Sll| layout is used. The \verb|Sll| at
the end specifies the layout for the result.

\subsubsection{Synthesising the in-place map function}

We can generalise our \verb|mapAdd1| to map arbitrary \verb|Int|
functions over a list and then redefine \verb|mapAdd1| using the new
\verb|map|.

\begin{lstlisting}
%generate mapAdd1 [Sll] Sll

data List := Nil | Cons Int List;

Sll : List >-> layout[x];
Sll (Nil) := emp;
Sll (Cons head tail) := x :-> head, (x+1) :-> tail, Sll tail;

map : (Int -> Int) -> List -> List;
map f (Nil) := Nil;
map f (Cons x xs) := Cons (instantiate [Int] Int f x) (map f xs);

add1 : Int -> Int;
add1 x := x + 1;

mapAdd1 : List -> List;
mapAdd1 xs := instantiate [Int -> Int, Sll] Sll map add1 xs;
\end{lstlisting}

The keyword \lstinline{instantiate} gives specific layouts to use for the types in function applications, \emph{e.g.}, if a function \verb|g| has type \verb|A -> B -> C -> D|, then
\lstinline{instantiate [L1, L2, L3] L4 g x y z} will use layout \verb|L1| for type \verb|A|, \verb|L2| for type \verb|B|, \verb|L3| for type \verb|C| and, finally, \verb|L4| for the result type \verb|D|, while applying the function \verb|g| to the three arguments \verb|x|, \verb|y| and \verb|z|.

This example makes use of \lstinline{instantiate} in two places.
In the first case where we have the call
\lstinline{instantiate [Int] Int f x}, the builtin \verb|Int| layout is used for both the input
and output. In this special case, the \verb|Int| layout shares a name with the \verb|Int| type
that it represents. This is necessary since \lstinline{instantiate} is used for all non-recursive calls that are not constructor applications.

In the second use, \lstinline{instantiate [Int -> Int, Sll] Sll map add1 xs}, we specify that
the second argument uses the \verb|Sll| layout for the \verb|List| type from \autoref{fig:List-def}. We also give \verb|Sll| as
the layout for the result of the call. Note that it is not necessary to use \verb|instantiate| for the recursive call to \verb|map|. This is because the appropriate layout
is inferred for recursive calls.

The type signature of \verb|mapAdd1| implies that it is layout polymorphic, as the type does not refer to any specific layout. It might be surprising that \verb|instantiate| is required in the body of \verb|mapAdd1| since the type signature
of \verb|mapAdd1| suggests that it is layout polymorphic and yet we must pick a specific \verb|List| layout when
we use \verb|instantiate| to call \verb|map|. This is because, in general, a call inside the body of some function \verb|fn| might use
any layout, even layouts that have no relation to the layouts that \verb|fn| is instantiated to. Finally, please note that our benchmarks shown in Figure~\ref{fig:size-comparison} include a more general version of \verb|mapAdd|.

\subsubsection{Guards}

While we have a pattern-matching construct at the top level of a function definition, we
have not seen a way to branch on a Boolean value so far. This is a feature that is
readily available at the level of \suslik, since the same conditional construct we
use to implement pattern matching can also use other Boolean expressions.

We can expose this in the functional language using a \textit{guard},
much like Haskell's guards. Suppose we want to write a specialised
filter-like function. Specifically, we want a function that filters
out all elements of a list that are less than 9. This is a specific
example where the \suslik specification is noticeably more difficult
to read. For a \suslik specification of this example, see
\autoref{fig:suslik-filter}

\begin{figure}
  \begin{lstlisting}
predicate filterLt9(loc x, loc r) {
| (x == 0) => { r == 0 ; emp }
| not (x == 0) && head < 9 =>
    { x :-> head ** (x+1) :-> tail ** [x,2] ** filterLt9(tail, r) }
| not (x == 0) && not (head < 9) =>
    { x :-> head ** (x+1) :-> tail ** [x,2] ** filterLt9(tail, y)
      ** r :-> head ** (r+1) :-> y ** [r,2] } }

void filterLt9(loc x1, loc r)
  { Sll(x1) ** r :-> 0 }
  { filterLt9(x1, r0) ** r :-> r0 }
{ ?? }
  \end{lstlisting}
  \caption{\suslik specification of \lstinline{filterLt9}, excluding \lstinline{Sll} which is given in \autoref{fig:List-def}}
  \label{fig:suslik-filter}
\end{figure}

On the other hand, an implementation of this in \tool is:

\begin{lstlisting}
%generate filterLt9 [Sll] Sll

filterLt9 : List -> List;
filterLt9 (Nil) := Nil;
filterLt9 (Cons head tail)
  | head < 9       := filterLt9 tail;
  | not (head < 9) := Cons head (filterLt9 tail);
\end{lstlisting}

\noindent
When translating a guarded function body, the translator takes the conjunction of the Boolean guard condition with
the condition for the pattern match. Finally, please note that our benchmarks shown in Figure~\ref{fig:size-comparison} include a more general version of \verb|filterLt|.

While the SuSLik version of \verb|filterLt9| requires working with pointers directly, the Pika version uses pattern matching and constructor application.
This allows the Pika code to work independent of the layout used.

\subsubsection{\texttt{if-then-else}}

Another feature that is common in functional languages is
\verb|if-then-else| expressions. This has a straightforward translation into
\suslik.
The \verb|if-then-else| construct corresponds to \suslik's C-like
ternary operator. We can use this feature to implement the \verb|even|
function which produces $1$ is the argument is even and $0$ otherwise.

\begin{lstlisting}
%generate even [Int] Int

even : Int -> Int;
even (n) := if (n % 2) == 0 then 1 else 0;
\end{lstlisting}

\subsubsection{Using multiple layouts}

To show the interaction between multiple algebraic data types, we write a
function that follows the left branches of a binary tree and collects the values
stored in those nodes into a list. This example demonstrates a binary tree
algebraic data type and a layout that corresponds to it.

\begin{lstlisting}
%generate leftList [TreeLayout] Sll

data Tree := Leaf | Node Int Tree Tree;

TreeLayout : Tree >-> layout[x];
TreeLayout (Leaf) := emp;
TreeLayout (Node payload left right) := x :-> payload, (x+1) :-> left,
  (x+2) :-> right, TreeLayout left, TreeLayout right;

leftList : Tree -> List;
leftList (Leaf) := Nil;
leftList (Node a b c) := Cons a (leftList b);
\end{lstlisting}

\subsubsection{Synthesising fold}
\label{sec:examples-fold}

A fold is a common kind of operation on a data structure in functional programming, where a binary function
is applied to the elements of a data structure to create a summary value. For example, if the binary function
is the addition function, it will give the sum of all the elements of the data structure. The
classic example of such a fold is a fold on a list. In this example, we will write a
right fold over a \verb|List|.

\begin{lstlisting}
%generate fold_List [Int, Sll] Int

fold_List : Int -> List -> Int;
fold_List z (Nil) := z;
fold_List z (Cons x xs) :=
  instantiate [Int, Int] Int f x (fold_List z xs);
\end{lstlisting}

\noindent
We will specifically look at the specialization where we use the addition function for \verb|f| so that we can focus on way that layouts are used in the translation. This sort of specialization corresponds to defunctionalization. The compiler produces the following \suslik specification for \verb|fold_List|:

\begin{lstlisting}
predicate fold_List(int i1, loc x, int i2) {
  | x == 0 => { i2 == i1 ; emp }
  | not (x == 0) => { (zz4 == i1) && ((zz5 == nxt13) &&
     (i2 == (h + b3))) ; [x,2] ** x :-> h ** (x + 1) :-> nxt13 **
     fold_List(zz4, zz5, b3) } }
\end{lstlisting}

The first two parameters of the \suslik predicate correspond to the two arguments of the Pika function. The final
parameter of the predicate, \verb|i2|, corresponds to the output of the Pika function. There are two cases:

\begin{enumerate}
  \item The \verb|x == 0| case corresponds to the \verb|Nil| case in \tool. In this case, the pure part of
    the assertion (to the left of the semicolon) requires that the output is equal to the first parameter.
    This is because the Pika function returns the first parameter in its \verb|Nil| case.

  \item The \verb|not (x == 0)| case corresponds to the \verb|Cons| case. First, let's look at the spatial part (this is
    everything to the \textit{right} of the semicolon). The pattern match destructures the \verb|Cons| into its head and tail.
    Likewise, in the spatial part of the SuSLik predicate, we require that \verb|x| points to \verb|h| (the head) and \verb|x + 1|
    points to \verb|nxt13| (the tail). We also recursively call the predicate on the tail. The pure part does two things:
    it introduces new names for things (these are used internally) and it requires that the output \verb|i2| is the sum of the head (\verb|h|)
    and the value obtained from the output of the recursive call (\verb|b3|).
\end{enumerate}




\section{Formal Semantics of \tool}
\label{sec:formal}

In this section, we have the following plan:

\begin{itemize}
  \item We define abstract machine semantics for executing a subset of \tool programs. This semantics
    is given by the big-step relation $\step$ which we define later.
  \item We define the translation from that subset of \tool into SSL. This translation is given by the
    function $\Tsem{e}{V,r}$ from \tool expressions into SSL propositions. The $r$ is a variable name to
    be used in the resulting proposition and $V$ is a collection of fresh names.
  \item We prove a soundness theorem. Given any well-typed expression $e$ and an abstract machine reduction producing
    the store-heap pair $\EndPair$, the SSL translation of $e$ should be satisfied by SSL model $\EndPair$. This is stated formally, and proven,
    in \autoref{thm:gen-soundness}.
\end{itemize}

This subset of \tool does not have guards or conditional expressions, but it does have pattern
matching. It also has the requirement that functions can only have one argument. Unlike the implementation, there
is no elaboration. As a result, every algebraic data type value must be lowered to a specific layout at every usage
and every function application must be explicitly instantiated with a layout for the argument and a layout for the result. We
also limit the available integer and Boolean operations for brevity.

The grammar for this subset is given in \autoref{fig:subset-grammar}. The grammar for types, layout definitions and algebraic data type definitions remain the same as before and are therefore omitted. $\instantiateS{A,B}{f}$ corresponds to \verb|instantiate [A] B f|. We also include another construct, $\lowerS{A}{C\; e_1 \cdots e_n}$. This says to use the specific layout $A$ for the given constructor application $C\; e_1 \cdots e_n$.

\begin{figure}
\begin{grammar}
  <i> ::= $\cdots$ $\mid$ -2 $\mid$ -1 $\mid$ 0 $\mid$ 1 $\mid$ 2 $\mid$ $\cdots$

  <b> ::= \texttt{true} $\mid$ \texttt{false}

  <e> ::= <var> $\mid$ <i> $\mid$ <b> $\mid$ <e> + <e> $\mid$ C $\overline{\langle e \rangle}$ $\mid$ $\instantiateS{A,B}{f}(\langle \textit{e} \rangle)$ $\mid$ $\lowerS{A}{\langle \textit{e} \rangle}$

  <fn-def> ::= $\overline{\langle \textit{fn-case} \rangle}$

  <fn-case> ::= f <pattern> := <e>
\end{grammar}
  \caption{Grammar for restricted \tool subset}
  \label{fig:subset-grammar}
\end{figure}

The semantics for SSL are largely derived~\cite{polikarpova:2019:suslik} from standard separation logic semantics.~\cite{rowe:2017:auto-cyclic-term}

\subsection{Overview of the Two Interpretations}

The soundness theorem will link the abstract machine semantics to the translation. In fact, the abstract machine
semantics and the translation are similar to each other. For the abstract machine we manipulate \textit{concrete} heaps, while
for the translation we generate \textit{symbolic} heaps.

Comparing the two further, there are two main points (beyond what we've already mentioned) where these two interpretations of \tool differ:


      \begin{enumerate}
        \item When we need to unfold a layout, how do we know which layout branch to choose?
        \item How do we translate function applications (including, but not limited to, recursive applications)?
      \end{enumerate}

First, consider the abstract machine semantics. In this case, we are able to choose which branch of a layout to use by evaluating the expression we are applying it
          to until the expression is reduced to a constructor value (where a ``constructor value'' is either a value or a constructor applied to constructor values). If the expression is well-typed, this will always be a constructor
          of the algebraic data type corresponding to the layout. The two rules that this applies to are {\sc AM-Lower} and
          {\sc AM-Instantiate}.
        To interpret a function application, we interpret its arguments and substitute the results into the body
          of the function. We then proceed to interpret the substituted function body. This process is performed by the {\sc AM-Instantiate} rule.

Next, consider the SSL translation. Here we can determine which layout branch to use by generating a Boolean condition that will be
          true if and only if the SSL proposition on the right-hand side of the branch holds for the heap. Note
          that we assume that the programmer-supplied layout definitions are \textit{injective} functions from
          algebraic data types to SSL assertions (up to bi-implication).
We can directly translate function applications into SSL inductive predicate applications. Since inductive predicates
          already allow for recursive applications, there is no special handling necessary for recursion.

%
%
%

After defining these interpretations, we show how the abstract machine relation $\step$ and the SSL interpretation function $\Tsem{e}{V,r}$ relate to each other by the Soundness \autoref{thm:gen-soundness}.

%
%
%

\subsection{Abstract Machine Semantics}
In this section, we will define an abstract machine semantics for \tool and relate this to
the standard semantics for SSL.

\subsubsection{Notation and Setup}


The set of values is $\Val = \mathbb{Z} \cup \mathbb{B} \cup \Loc$. Each of these three sets
is disjoint from the other two. In particular, note that $\Loc$ and $\mathbb{Z}$ are disjoint.

There is also a set of \tool values \FsVal. This includes all the elements of \Val{}, but also includes
``constructor values'' given by the rules in \autoref{fig:FsVal-rules}.
In addition to the store and heap of standard SSL semantics, the abstract machine semantics uses an \FsStore. This is a partial function from locations
to \tool values: $\FsStore = \Loc \rightharpoonup \FsVal$. The primary purpose of this is to recover constructor values when given a location.

The general format of the transition relation is $(e, \sigma, h, \mathcal{F}) \step (v, \sigma', h', \mathcal{F}', r)$, where the expression $e$ results
in the store being updated from $\sigma$ to $\sigma'$, the heap being updated from $h$ to $h'$, $v$ is the \Val{} obtained by
evaluating $e$, the initial and final $\FsStore$ are $\mathcal{F}$ and $\mathcal{F}'$ and the result is stored in variable $r$. We assume that there is a global environment $\Sigma$ which contains all layout definition equations and function definition equations.
%
%



Given a heap $h$ and a heap layout $H$, we will make use of the
notation $h \cdot [H]$. This extends the heap with the location
assignments given in $H$. We say that the layout body $H$ is
\textit{acting} on the heap $h$. This is defined in
\autoref{fig:layout-act}. The intuition for this is that $h$ gets
updated using the symbolic heap description in $H$. For example,
$\emptyset \cdot [a \pointsto 7]$
will contain only the value $7$ at the location
$a$. It is assumed that $H$ does not have any variables on the
right-hand side of $\pointsto$.

\subsubsection{Abstract Machine Rules}

The abstract machine semantics provides big-step operational semantics
for evaluating \tool expressions on a heap machine. Its rules, given by
\autoref{fig:abs-machine}, make use of standard SSL models:

\begin{align*}
  &\text{Model} &\mathcal{M} ::= (\sigma, h)\\
  &\text{Store} &\sigma : \text{Var} \rightharpoonup \text{Val}\\
  &\text{Heap} &h : \text{Loc} \rightharpoonup \text{Val}\\
\end{align*}

\noindent
Note that a compound expression, consisting of multiple
subexpressions, uses \emph{disjoint} parts of the heap for each
subexpression. This can be seen in the {\sc AM-Add}, {\sc AM-Lower}
and {\sc AM-Instantiate} rules, which is essential for the proof of
Soundness \autoref{thm:gen-soundness}.

\begin{figure}[thbp]
  \[
  \begin{array}{c}
    \labinfer[\FsVal-Base]{x \in \FsVal}
      {x \in \Val}
    ~~~
    \labinfer[\FsVal-Constr]{(C\; x_1 \cdots x_n) \in \FsVal}
      {x_1 \in \FsVal & \cdots & x_n \in \FsVal}
  \end{array}
  \]
    \caption{\FsVal{} judgment rules}
    \label{fig:FsVal-rules}

    \bigskip

  \resizebox{\textwidth}{!}
{
\[
\begin{array}{c}
  \labinfer[AM-Int]{(i, \sigma, \emptyheap, \mathcal{F}) \step (i, \sigma', \emptyheap, \mathcal{F}, r) }
    {\freshVar{r} & i \in \mathbb{Z} & \sigma' = \sigma \cup \{ (r, i) \}}
  ~~~
  \labinfer[AM-Bool]{(b, \sigma, \emptyheap, \mathcal{F}) \step (b, \sigma', \emptyheap, \mathcal{F}, r) }
    {\freshVar{r} & b \in \mathbb{B} & \sigma' = \sigma \cup \{ (r, b) \}}
  \\\\
  \labinfer[AM-Var-Base]{(v, \sigma, \emptyheap, \mathcal{F}) \step (\sigma(v), \sigma, \emptyheap, \mathcal{F}, v)}
    {v \in dom(\sigma) & \sigma(v) \not\in \Loc}
  \\\\
  \labinfer[AM-Var-Loc]{(v, \sigma, \emptyheap, \mathcal{F}) \step (\mathcal{F}(\sigma(v)), \sigma, \emptyheap, \mathcal{F}, v)}
    {v \in dom(\sigma) & \sigma(v) \in \Loc}
  \\\\
  \labinfer[AM-Add]{(x + y, \sigma, h, \mathcal{F}) \step (z, \sigma', h', \mathcal{F}, r)}
    {\begin{gathered}
      (x, \sigma, \mathcal{F}, h_1) \step (x', \sigma_x, h_1', \mathcal{F}, v_x)
      ~~~ (y, \sigma_x, \mathcal{F}, h_2) \step (y', \sigma_y, h_2', \mathcal{F}, v_y)
    \\  \freshVar{r}
    ~~~ h = h_1 \circ h_2
    ~~~ h' = h_1' \circ h_2'
    ~~~ z = x' + y'
    ~~~ \sigma' = \sigma_y \cup \{ (r, z) \}
     \end{gathered}
    }
  \\\\
  \labinfer[AM-Lower]{(\lowerS{A}{e}, \sigma_0, h, \mathcal{F}) \step ((C\; e_1' \cdots e_n'), \sigma', h' \cdot [H'], \mathcal{F}', r)}
     {\begin{gathered}
          (A[x]\; (C\; a_1 \cdots a_n) := H) \in \Sigma
       ~~~ (e, \sigma_0, h_0) \step (C\; e_1 \cdots e_n, \sigma_1, h_1, y_1)
       \\  (e_i, \sigma_i, h_i) \step (e_i', \sigma_{i+1}, h_{i}', v_i)\; \textnormal{for each $1 \le i \le n$}
       \\ h' = h_1' \circ h_2' \circ \cdots \circ h_n'
       ~~~ h = h_0 \circ h_1 \circ \cdots \circ h_n
       ~~~ \sigma' = \sigma_{n+1} \cup \{ (r, \ell) \}
       \\  \freshLoc{\ell}
       ~~~ \freshVar{r}
       \\  H' = H[x := \ell][a_1 := \sigma_2(v_1)][a_2 := \sigma_3(v_2)]\cdots[a_n := \sigma_{n+1}(v_n)]
       \\  \mathcal{F}' = \mathcal{F} \cup \{ (\ell, (C\; e_1' \cdots e_n')) \}
     \end{gathered}
    }
  \\\\
  \labinfer[AM-Instantiate]{(\instantiateS{A,B}{f}(e), \sigma, \mathcal{F}, h) \step (e_f', \sigma', h'', \mathcal{F}', r)}
    {\begin{gathered}
        (A[x]\; (C\; a) := H) \in \Sigma
    ~~~ (f\; (C\; b) := e_f) \in \Sigma
    \\  (e, \sigma, h) \step (C\; e_1, \sigma_1, h_1, y)
    \\  (e_1, \sigma_1, h_1) \step (e_1', \sigma_2, h_2, r)
    \\ \freshLoc{\ell}
    ~~~ \freshVar{r}
      ~~~ \freshVar{y}
    \\  H' = H[x := \ell][a := e_1']
    ~~~  h' = h_1 \cdot [H']
    ~~~ \sigma' = \sigma_f \cup \{ (r, \ell) \}
    \\  \mathcal{F}' = \mathcal{F} \cup \{ (\ell, (C\; e_1')) \}
    \\  (\lowerS{B}{e_f[b := y]}, \sigma_2, h') \step (e_f', \sigma_f, h'', r)
     \end{gathered}
    }
\end{array}
\]
}


  \caption{Abstract machine semantics rules}
  \label{fig:abs-machine}

  \bigskip
  \[
  \begin{array}{c}
    \labinfer[L-Emp]{h \cdot [\emp] = h}
      {}
    ~~~
    \labinfer[L-PointsTo]{h \cdot [\ell \pointsto a, H] = h'[\ell \mapsto a]}
      {h' = h \cdot [H] & a \in \Val}
    \\\\
    \labinfer[L-Apply]{h \cdot [A[x](e), H] = h \cdot [H]}
      {e \in \Val}
  \end{array}
  \]
    \caption{Rules for layout bodies acting on heaps}
    \label{fig:layout-act}
\end{figure}


\subsection{Translating \tool Specifications into SSL}

We will define two translations: One from \tool \textit{expressions} into SSL propositions and the other from
\tool \textit{definitions} into SSL inductive predicate definitions. We start with the former.

\subsubsection{Translating Expressions}

In the rules given in \autoref{fig:expr-rules}, the notation $\mathcal{I}_{A,B}(f)$ gives the name of the inductive predicate that the \tool
function $f$ translates to when it is instantiated to the layouts $A$ and $B$.

We start by defining the translation rules for expressions. We use these translation rules in \autoref{thm:ssem-fn} to
define a translation function $\Tsem{\cdot}{\VarSet,r}$. Then, we will define translation rules for function definitions. The translation relation for expressions has the form $(e, \VarSet) \tstep (p, s, \VarSet', v)$,
where $p$ and $s$ are the pure part and spatial part (respectively) of an SSL assertion and $\VarSet,\VarSet' \in \mathcal{P}(\Var)$.

The rules can be thought of as being in two groups:

\begin{enumerate}
  \item Rules for base type expressions, such as {\sc S-Lit} and {\sc S-Add}.
  \item Rules for using layouts to translate expressions whose types involve algebraic data types. Examples
    include {\sc S-Lower-Constr} and {\sc S-Inst-Inst}.
\end{enumerate}

In the first group, consider {\sc S-Add}. In the result of the translation, we've included $v == v_1 + v_2$ in the list of conjuncts
in the pure part. Here, $v_1$ and $v_2$ are the results of the two subexpressions in the addition. In the pure part,
we also include the pure parts of the two subexpressions as conjuncts. These are $p_1$ and $p_2$. The spatial part
of the translation consists of the spatial parts of the two subexpressions, $s_1$ and $s_2$.

Now, in the second group, consider {\sc S-Lower-Constr}. This translates a Pika constructor application expression using
a specific layout (which is provided by using the $\lowerS{-}{-}$ construct). It takes the specific branch of the layout
corresponding to the constructor in question and puts the right-hand side of that branch into the spatial part, after applying
the appropriate substitutions for the arguments given to the constructor in the application. The right-hand side of the layout
branch is $H$ and, after the substitution, it is called $H'$.

The {\sc S-Inst-Inst} rule is used to translate a function application being applied to the result of another function application, given
particular layouts for each application. In \suslik, it does not make sense to directly apply a predicate to another predicate application. Therefore, we must do an ANF-like
translation, where the result does not have ``compound'' applications like this. This translation is exactly what {\sc S-Inst-Inst} is doing.

\begin{figure}[thbp]
    \resizebox{\textwidth}{!}
{
\[
\begin{array}{c}
  \labinfer[S-Int]{(i, \VarSet) \tstep (v == i, \emp, \VarSet \cup \{v\}, v)}
    {i \in \mathbb{Z} &
     \freshVar{v}}
  \\\\
  \labinfer[S-Bool]{(b, \VarSet) \tstep (v == b, \emp, \VarSet \cup \{v\}, v)}
    {b \in \mathbb{B} &
     \freshVar{v}
    }
  \\\\
  \labinfer[S-Var]{(v, \VarSet) \tstep (\texttt{true}, \emp, \VarSet, v)}
    {v \in \Var}
  \\\\
  \ruleSAdd
  \\\\
  \labinfer[S-Lower-Var]{(\lowerS{A}{v}, \VarSet) \tstep (\texttt{true}, A(v), \VarSet \cup \{v\}, v)}
    {\begin{gathered}
      v \in \Var
     \end{gathered}
    }
  \\\\
  \labinfer[S-Lower-Constr]{
        (\lowerS{A}{C\; e_1 \cdots e_n}, \VarSet_1)
          \tstep (p_1 \land \cdots \land p_n, H' \sep s_1 \sep \cdots \sep s_n, \VarSet', x)}
    {\begin{gathered}
          (A[x]\; (C\; a_1 \cdots a_n) := H) \in \Sigma
      \\  (e_i, \VarSet_i) \tstep (p_i, s_i, \VarSet_{i+1}, v_i)\; \textnormal{for each $1 \le i \le n$}
      \\  \freshVar{v}
      \\  \VarSet' = \VarSet_{n+1} \cup \{x\}
      ~~~ H' = H[a_1 := v_1]\cdots[a_n := v_n]
     \end{gathered}
    }
  \\\\
  \labinfer[S-Inst-Var]{(\instantiateS{A,B}{f}(v), \VarSet) \tstep (\texttt{true}, \mathcal{I}_{A,B}(f)(v, r), \VarSet \cup \{r\}, r)}
    {v \in \VarSet & \freshVar{r}}
  \\\\
  \labinfer[S-Inst-Constr]{
        (\instantiateS{A,B}{f}(C\; e_1 \cdots e_n), \VarSet_1)
                                     \tstep (p \land p_1 \land \cdots \land p_n, s \sep s_1 \sep \cdots \sep s_n, \VarSet', r)
      }
    {\begin{gathered}
          (A[x]\; (C\; a_1 \cdots a_n) := H) \in \Sigma
      \\  (e_i, \VarSet_i) \tstep (p_i, s_i, \VarSet_{i+1}, v_i)\; \textnormal{for each $1 \le i \le n$}
      \\  \freshVar{x}
      \\  \VarSet' = \VarSet_{n+1} \cup \{x\}
      ~~~ H' = H[a_1 := v_1]\cdots[a_n := v_n]
      \\
          (f\; (C\; b_1 \cdots b_n) := e_f) \in \Sigma
      ~~~ e_f' = e_f[b_1 := v_1]\cdots[b_n := v_n]
      \\  (\lowerS{B}{e_f'}, \VarSet_{n+1}) \tstep (p, s, \VarSet', r)
     \end{gathered}
    }
  \\\\
  \labinfer[S-Inst-Inst]{
      (\instantiateS{B,C}{f}(\instantiateS{A,B}{g}(e)), \VarSet) \tstep (p_1 \land p_2, s_1 \sep s_2, \VarSet_2, r_2)
    }
    {\begin{gathered}
      (\instantiateS{A,B}{g}(e), \VarSet) \tstep (p_1, s_1, \VarSet_1, r_1)
      \\ (\instantiateS{B,C}{f}(r_1), \VarSet_1) \tstep (p_2, s_2, \VarSet_2, r_2)
     \end{gathered}
    }
\end{array}
\]
}
\caption{Expression Translation Rules}
\label{fig:expr-rules}
\bigskip
\[
  \begin{array}{c}
    \labinfer[FnDef]{(f\; (C\; a_1 \cdots a_n) := e) \defstep{A,B} (\mathcal{I}_{A,B}(f)(x, r) : c \Ra \{p ; s\})}
      {\begin{gathered}
            \VarSet = \{ v_1 , \cdots , v_n \}\; \textnormal{where $v_1, \cdots, v_n$ are distinct variables}
        \\  r \in \Var
        ~~~ r \not\in \VarSet
        ~~~ c = \cond(A, C, x)
        ~~~ p_1 \cdots p_n~\textnormal{fresh}
        \\  (p, s) = \Tsem{\instantiateS{A,B}{f}(C\; p_1 \cdots p_n)}{\VarSet, r}
       \end{gathered}
      }
  \end{array}
  \]
  \caption{Translation rule for function definitions}
  \label{fig:FnDef-rule}

  \end{figure}

\begin{figure}[t]
  
  \end{figure}

\begin{lemma}[$\Tsem{\cdot}{}$ function] \label{thm:ssem-fn} $(\cdot, \VarSet) \tstep (\cdot, \cdot, \cdot, r)$ is a computable function $\Expr \rightarrow (\Pure \times \Spatial \times \mathcal{P}(\Var))$, given fixed
  $\VarSet$ and $r$ where $r \not\in \VarSet$.

By throwing away the third element of the tuple in the codomain, we obtain a function $\Expr \rightarrow (\Pure \times \Spatial)$ from expressions to
SSL propositions.

  Call this function $\Tsem{\cdot}{\VarSet, r}$. That is, we define the function as follows where $r \not\in \VarSet$:
  \[
    \Tsem{e}{\VarSet, r} = (p ; s) \iff (e, \VarSet) \tstep (p, s, \VarSet', r)\; \textnormal{for some $\VarSet'$}
  \]
\end{lemma}

We highlight the computability of this function to emphasise the fact
that it can be used directly in an implementation of this subset of
\tool.

\subsubsection{Translating Function Definitions}

The next step is to define the translation for \tool function definitions. In order to do this, we must first figure out how to determine the
appropriate layout branch to use when unfolding a layout, a problem we highlighted earlier. Once this is accomplished, the rest of
the translation can be defined. When this problem was solved for the abstract machine semantics,
it was possible to simply evaluate the \tool expression until a constructor application expression was reached. From there, it is possible to just look at the constructor name and match it against the appropriate layout branch.

For the translation, however, we do not have the luxury of being able to evaluate expressions. Instead, we must instead rely on the fact
that, in SSL, a ``pure'' (Boolean) condition can determine which inductive predicate branch to use. The question becomes: \textit{Given an algebraic
data type and a layout for that ADT, how do we
generate an appropriate Boolean condition for a given constructor for the ADT}?

The solution is to find a Boolean condition which, given that the inductive predicate holds, is true \textit{if and only if} the layout branch corresponding to that
ADT constructor is satisfiable. In more detail, to define the branches of an inductive predicate $\mathcal{I}_{A}(x)$, given an ADT $\alpha$, a constructor $C : \beta_1 \rightarrow \beta_2 \rightarrow \cdots \rightarrow \beta_n \rightarrow \alpha$, a layout $A : \alpha \monic \layout[x]$ with a branch $A[x]\; (C\; a_1 \cdots a_n) := H$ and given that $\mathcal{I}_{A}(x)$ holds, find a Boolean expression $b$ with one free variable $x$ such that $b \iff \exists \sigma, h.\; (\sigma, h) \models H$.

\begin{lemma}[$\cond$ function]
  \label{thm:cond}
  There is a computable function $\cond(\cdot,\cdot)$ that takes in any layout $A : \alpha \monic \layout[x]$ with a branch $A[x]\; (C\; a_1 \cdots a_n) := H$ for a given constructor
  $C : \beta_1 \rightarrow \beta_2 \rightarrow \cdots \rightarrow \beta_n \rightarrow \alpha$ and it produces a Boolean expression with one free variable $x$ such that the following holds under the assumption that $\mathcal{I}_{A}(x)$ holds.
  \[
    \cond(A, C) \iff \exists \sigma, h.\; (\sigma, h) \models H
  \]
  Here, $\mathcal{I}_{A}(x)$ is the name of the generated inductive
  predicate corresponding to the layout~$A$.
\end{lemma}

With this function in hand, we are now ready to define the translation for \tool function definitions. This
definition is in \autoref{fig:FnDef-rule}. In this rule, the fresh variables $p_1, \cdots, p_n$ will be substituted for $a_1, \cdots, a_n$.


\subsection{Typing Rules}

Typing rules for \tool expressions are given in
\autoref{fig:typing-rules}. These rules differ from standard typing
rules for a functional language due to the existence of layouts and
their associated constructors, like \verb|instantiate| and
\verb|lower|. If an expression is well-typed, then each use of
\verb|instantiate| and \verb|lower| only uses layouts together with
the ADT that they are defined for.

The rules also make use of a \textit{concreteness judgment}. The rules
for this judgment are given in \autoref{fig:concreteness-rules}. The
intuition of this judgment is that a type is ``concrete'' iff values
of that type can be directly represented in the heap machine
semantics. For example, an ADT type is \textit{not} concrete because a
layout has not been specified. However, once a particular layout is
specified for the ADT type, it becomes concrete. Base types, like
\verb|Int|, are also concrete.

Rules for the ensuring that global definitions are well-typed are
given in \autoref{fig:globals-typing-rules}. In this figure, $\Delta$
is the set of all (global) constructor type definitions.

\begin{figure}[bth]
\[
\begin{array}{c}
  \labinfer[T-Int]{\Gamma \vdash i : \Int}
    {i \in \mathbb{Z}}
  ~~~
  \labinfer[T-Bool]{\Gamma \vdash b : \Bool}
    {b \in \mathbb{B}}
  ~~~
  \labinfer[T-Var]{\Gamma \vdash v : \alpha}
    {(v : \alpha) \in \Gamma}
  \\\\
  \labinfer[T-Fn-Global]{\Gamma \vdash f : \alpha \rightarrow \beta}
    {(f : \alpha \rightarrow \beta) \in \Sigma}
  \\\\
  \labinfer[T-Add]{\Gamma \vdash x + y : \Int}
    {\Gamma \vdash x : \Int & \Gamma \vdash y : \Int}
  \\\\
  \labinfer[T-Lower-Var]{\Gamma \vdash \lowerS{A}{v} : A}
    {(v : \alpha) \in \Gamma & (A : \alpha \monic \layout[x]) \in \Sigma}
  \\\\
  \labinfer[T-Lower-Constr]{\Gamma \vdash \lowerS{A}{C\; e_1 \cdots e_n} : B}
    {\begin{gathered}
      (C : \alpha_1 \rightarrow \cdots \rightarrow \alpha_n \rightarrow \beta) \in \Sigma
      ~~~ (B : \beta \monic \layout[x]) \in \Sigma
      \\  \Gamma \vdash e_i \concrete{\alpha_i}\; \textnormal{for each $i$ with $1 \le i \le n$}
     \end{gathered}
    }
  \\\\
  \labinfer[T-Instantiate]{\Gamma \vdash \instantiateS{A,B}{f}(e) : B}
    {\begin{gathered}
          (A : \alpha \monic \layout[x]) \in \Sigma
      ~~~ (B : \beta \monic \layout[y]) \in \Sigma
      \\  \Gamma \vdash f : \alpha \rightarrow \beta
      ~~~ \Gamma \vdash e : A
     \end{gathered}
    }
  \\\\
  \labinfer[T-Constr]{\Gamma \vdash C\; e_1 \cdots e_n : \beta}
    {(C : \alpha_1 \rightarrow \cdots \rightarrow \alpha_n \rightarrow \beta) \in \Sigma &
     \Gamma \vdash e_i : \alpha_i\; \textnormal{for each $i$ with $1 \le i \le n$}
    }
\end{array}
\]
  \caption{Typing rules}
  \label{fig:typing-rules}
\end{figure}

\begin{figure}[b]
\centering
\begin{subfigure}{0.49\textwidth}
\begin{minipage}{1.0\linewidth}
{\footnotesize{
\[
  \begin{array}{c}
    \labinfer[C-Int]{e \concrete{\Int}}
      {\Gamma \vdash e : \Int}
    ~~~
    \labinfer[C-Bool]{e \concrete{\Bool}}
      {\Gamma \vdash e : \Bool}
    \\\\
    \labinfer[C-Layout]{e \concrete{\alpha}}
      {(A : \alpha \monic \layout[x]) \in \Sigma &
       \Gamma \vdash e : A
      }
  \end{array}
\]
}}
\end{minipage}
  \caption{Concreteness judgment rules}
  \label{fig:concreteness-rules}
\end{subfigure}
\begin{subfigure}{0.49\textwidth}
\begin{minipage}{1.0\linewidth}
{\footnotesize{
\[
\begin{array}{c}
  \labinfer[G-Fn]{(f\; (C\; b_1 \cdots b_n) := e) \Rightarrow f : \beta \rightarrow \gamma}
    {(C : \alpha_1 \rightarrow \alpha_2 \rightarrow \cdots \rightarrow \alpha_n \rightarrow \beta) \in \Delta\\
    b_1 : \alpha_1, b_2 : \alpha_2, \cdots, b_n : \alpha_n \vdash e : \gamma
    }
\end{array}
\]
}}
\end{minipage}
  \caption{Global definition typing}
  \label{fig:globals-typing-rules}
\end{subfigure}
\caption{Rules for concreteness judgement and typing global definitions}
\label{fig:meh}
\end{figure}


\subsection{From \tool to SLL Specifications: Soundness of the Translation}
\label{sec:soundness}

We want to show that our abstract machine semantics and our SSL
translation fit together. In particular, our abstract machine
semantics should generate models that satisfy the separation logic
propositions given by our SSL translation.
Figure~\ref{fig:soundness-diagram} gives a high-level overview of how
these pieces fit together. We will give a more specific description of
this in Theorem~\ref{thm:gen-soundness}.


\begin{figure}
\[\begin{tikzcd}
	{\textnormal{\tool}} &&&&& {\textnormal{SSL propositions}} \\
	\\
	\\
	{\textnormal{Abstract machine semantics}}
	\arrow["{\textnormal{translation}}", from=1-1, to=1-6]
	\arrow["{\textnormal{interpretation}}"', from=1-1, to=4-1]
	\arrow["\models"', from=4-1, to=1-6]
\end{tikzcd}\]
  \caption{The relationship between the two \tool{} semantics given by the soundness theorem}
  \label{fig:soundness-diagram}
\end{figure}

\begin{restatable}[Soundness]{theorem}{soundnessThm} \label{thm:gen-soundness}
  For any well-typed expression $e$, if $\Tsem{e}{\VarSet, r}$ is
  satisfiable for $\VarSet = \dom(\EndSigma)$ and $(e, \StartSigma, \Start{h}, \mathcal{F}) \step (e', \EndSigma, \End{h'}, \mathcal{F}', r)$,
          then $\EndPair \models \Tsem{e}{\VarSet, r}$. 
  That is, given an expression $e$ with a satisfiable SSL translation, any heap machine state that $e$ transitions to (by the abstract machine semantics) will
  be a model for the SSL translation of $e$ (\cf~\autoref{fig:soundness-diagram}).
\end{restatable}

\begin{proof}
  See the Appendices
  in the extended version of the paper.
\end{proof}

The fact that, at the top level, we only translate function definitions suggests an additional theorem. We want to specifically
show that any possible function application is sound, in the sense just described. This immediately follows from Theorem~\ref{thm:gen-soundness}.

Abbreviating $\instantiateS{A,B}{f}$ as $f_{A,B}$, we can state the
following theorem:

\begin{theorem}[Application soundness] \label{thm:fn-soundness}
  For any well-typed function application $f_{A,B}(e)$, if
  $\Tsem{f_{A,B}(e)}{\VarSet, r}$ is satisfiable for $V = \dom(\EndSigma)$ and $(f_{A,B}(e), \StartSigma, \Start{h}, \mathcal{F}) \step (e', \EndSigma, \End{h'}, \mathcal{F}', r)$,
          then $\EndPair \models \Tsem{f_{A,B}(e)}{\VarSet, r}$.
\end{theorem}

\begin{proof}
  This follows immediately from Theorem~\ref{thm:gen-soundness}.
\end{proof}


\section{Extensions of \suslik}
\label{sec:extensions}

We have shown the translation from the functional specifications into
SSL specifications. However, some of the SSL specifications are not
supported in the original \suslik and existing variants.
In this section, we show how to extend the \suslik to support more
features to make the whole thing work. We will show the extensions on
the following three aspects:

\begin{itemize}
  \item How to describe and call an existing function within SSL predicates.
  \item How to make the result of one function call as the input of another function call.
  \item How to synthesise programs with inductive predicates without the help of pure theory.
\end{itemize}

\subsection{Function Predicates}
\label{sec:funcPred}

Without any modification upon the implementation, we find the SSL
predicate with some restrictions can be used to describe function
relations other than data structures (named function predicates). The
definition of \textbf{function predicates} is as follows:

\begin{definition}[Function Predicates]
    \label{def:funcPred}
    \normalfont
    Given any non-higher-order n-ary function $f$\lstinline{(x1, ..., xn)} in the functional language, the function predicate to synthesise $f$ has the following format:
    \begin{lstlisting}[language=SynLang]
    predicate predf(T x1, ... ,T xn, T output){...}
    \end{lstlisting}
    where \lstinline{T} $\in$ \{\lstinline{loc}, \lstinline{int}\}.
    The type of \lstinline{xi} (and \lstinline{output}) is decided by
    the type of $f$. If it is an integer in $f$, then its type is
    \lstinline{int}; otherwise, it is \lstinline{loc} (for any data
    structure in \tool).
\end{definition}

Since the input of the whole workflow is functional programs, the
``output'' in the definition is to provide another location for the
output of the function. And the specification to synthesise function
$f$ should have the following format:

\begin{lstlisting}[language=SynLang]
void f(loc x1, ... ,loc xn)
{x1 :-> v1 ** x2 :-> l2 ** sll(l2) ** ... ** xn :-> vn ** output :-> 0}
{x1 :-> v1 ** x2 :-> l2 ** ... ** xn :-> vn ** output :-> output0 ** 
 predf(v1, l2, ..., vn, output0)}
\end{lstlisting}

\subsection{SSL Rules for \func Structure}

As we show in previous examples, the reason we can have \func structure is that the points-to structure in the post-condition is always eliminated after write operations. For example, the in-placed \lstinline{inc1} functions specification is satisfied via the \writer (\autoref{fig:write}) on the location.

\begin{lstlisting}[language=SynLang]
void inc_y(loc y, loc x)
{x :-> vx ** y :-> vy}
{x :-> vx + vy ** y :-> vy}
\end{lstlisting}

\begin{figure}[t]
    \centering
    \begin{mathpar}
      \inferrule[\writer]
      {
      \mcode{\vars{e} \subseteq \env}
      \\
      \mcode{e \neq e'}
      \\
      \trans{\mcode{{\asn{\phi; \ispointsto{x}{e}{} \osep P}}}}
              {\mcode{\asn{\psi; \ispointsto{x}{e}{} \osep Q}}} {\mcode{\prog}}
      }
      {
      \mcode{
      \trans{\asn{\phi; {x} \pts e' \osep P}}
      {\asn{\psi; {x} \pts e \osep Q}}
      {{\deref{x} = e\ ;\ \prog}}
      }
      }
    \end{mathpar}
    \bigskip
    \begin{mathpar}
      \inferrule[\funcwrite]
        {
        \mcode{\forall i \in [1, n], \vars{e_i} \subseteq \env}
        \\
        \trans{\mcode{{\asn{\phi; P}}}}
              {\mcode{\asn{\psi; Q}}} {\mcode{\prog}}
        }
        {
        \mcode{\trans{\asn{\phi; {x} \pts e \osep P}}
        {\asn{\psi; \func\ f(e_1,\ldots,e_n,{x})\osep Q}}
        {{f(e_1,\ldots,e_n,{x})\ ;\ \prog}}
        }
        }
    \end{mathpar}
    
    \caption{The \writer and new \funcwrite rules in SSL}
    \label{fig:write}
\end{figure}



The core insight of \func structure is: since the function synthesised
by function predicate behaves like the pure function, it is the same
as the \writer rule in the sense that only the output location is
modified. Thus, we add the new \funcwrite rule into the zoo of SSL
rules (see \autoref{fig:write}).
To make the \func structure correctly equal to some ``write''
operation, the following restrictions should hold, which are achieved
by the translation:

\begin{itemize}
    \item If \lstinline[language=SynLang]{func f(x1, ..., xn, output)} appears in a post-condition, then no write rule can be applied to any \lstinline[language = SynLang]{xi}. This is to avoid the ambiguity of the \func.
    \item The type of function f is consistent.
\end{itemize}

Note that based on the setting of the function predicate, the parameters of the function call are pointers, while the parameters of the function predicate are content to which pointers point. Furthermore, we have the \func generated from function predicates and with the format defined in \autoref{sec:funcPred}. As a result, the equivalent original SSL that duplicates points-to of one location is not a problem, since they can be merged as one.

\subsection{Temporary Location for the Sequential Application}

Though richly expressive, SSL has difficulty in expressing the sequential application of functions. For example, given the \func structure available, the following function is not expressible within one function predicate:

\begin{lstlisting}[language=SynLang]
f x y = g (h x) y
\end{lstlisting}

If we attempt to express it, we will have the following part in the predicate:

\begin{lstlisting}[language=SynLang]
predicate f(loc x, loc y, loc output)
{... ** func h(x, houtput) ** func g(houtput, y, output)}
\end{lstlisting}

However, \lstinline{houtput} is not a location in the pre-condition, which is not allowed in SSL. Thus, we introduce a new keyword \lstinline{temp} to denote the temporary location for the sequential application. The new definition of \lstinline{func} is as follows:

\begin{lstlisting}[language=SynLang]
predicate f(loc x, loc y, loc output)
{... ** temp houtput ** func h(x, houtput) ** func g(houtput, y, output)}
\end{lstlisting}

Roughly speaking, the \lstinline{temp} structure will help to allocate a new location for the output of the first function, and then use it as the input of the second function. After all appearances of \lstinline{houtput} is eliminated, we will deallocate the location.

Note that the temporary variable is possible to appear in two different structures: recursive function predicates or \func call. The reason we don't need to consider the basic arithmetic operations is that the integer will be directly used as the predicate parameter, instead of the location as the parameter. For example, the sum of a list can be expressed as:

\begin{lstlisting}[language=SynLang]
predicate sum(loc l, int output){
| l == 0 => {output == 0; emp}
| l != 0 => {output == output1 + v; [l, 2] ** l :-> v ** l + 1 :-> lnxt ** sum(lnxt, output1)} }
\end{lstlisting}

Such sequential application is common in functional programming, especially in the recursive function. For example, it is not elegant to flatten a list of lists without the sequential application. 

\begin{lstlisting}[language=SynLang]
    flatten :: [[a]] -> [a]
    flatten [] = []
    flatten (x:xs) = x ++ flatten xs
\end{lstlisting}

We can express this function, but with some strange structure to store all temporary lists.

\begin{lstlisting}[language=SynLang]
predicate flatten(loc x, loc output){
| x == 0 => {output :-> 0}
| x != 0 => {[x, 2] ** x :-> x0 ** sll(x0) ** x + 1 :-> xnxt **
[output, 2] ** func append(x, outputnxt, output) ** output + 1 :-> outputnxt **
flatten(xnxt, outputnxt)} }
\end{lstlisting}

With such a function predicate, though we can synthesise the function
whose result stored in \lstinline{output} is the flattened list, the
list \lstinline{output} is containing a lot of intermediate values,
which is neither consistent with the definition in the source language
nor space efficient.

\begin{figure}[t]
  \centering
  \begin{mathpar}
    \inferrule[\tempfuncalloc]
      {
      \trans{\mcode{{\asn{\phi; x \pts a \osep P}}}}
      {\asn{\psi; \func\ f(e_1,\ldots,e_n,{x})\osep temp (x, 1) \osep Q}} {\mcode{\prog}}
      }
      {
      \mcode{
      \trans{\asn{\phi; P}}
      {\asn{\psi; \func\ f(e_1,\ldots,e_n,{x})\osep temp (x, 0) \osep Q}}
      {{let\ x\ =\ malloc(1)\ ;\ \prog}}
      }
      }
  \end{mathpar}
  \bigskip
  \begin{mathpar}
    \inferrule[\tempfuncfree]
      {
      \trans{\mcode{{\asn{\phi; P}}}}
      {\asn{\psi; Q}} {\mcode{\prog}}
      }
      {
      \mcode{\not\exists x\in Q\ \wedge
      \trans{\asn{\phi; P}}
      {\asn{\psi;  temp (x, 1) \osep Q}}
      {}
      }\\
      \mcode{let\ x0\ =\ *x\ ;\ type\_free(x0);\ free(x);\ \prog}
      }
  \end{mathpar}
  
  \caption{New allocating and deallocating rule for \lstinline{temp} in SSL }
  \label{fig:newalloc}
\end{figure}

  
  

The new rules consist of allocating  and deallocating rules (\autoref{fig:newalloc}). Based on the definition of the \func structure and the function predicate, the allocated locations are different, where the \lstinline{temp} location for \func is directly used; while the \lstinline{temp} location for function predicate should allocate a new location for function predicates. As for the deallocation, not only the \lstinline{temp} location(s) but also the content they point to should be deallocated. That is the reason we have the \lstinline{type_free} function, which is syntax sugar to deallocate the content of a location based on the type information. For example, if the type of the location is \lstinline{tree}, then the \lstinline{type_free} will deallocate the content of the location via \lstinline{tree_free} function, which is synthesised based on the SSL predicate \lstinline{tree} as follows.
\begin{lstlisting}[language=SynLang]
void tree_free(loc x)
  {tree(x)}
  {emp}
\end{lstlisting}
Specifically, if the location contains the value with type \lstinline{int}, then the \lstinline{type_free} will do nothing.
Thus, the function predicate with \lstinline{temp} is much better, in the sense that no extra space is used, and the synthesised function is consistent with the source language.

\begin{lstlisting}[language=SynLang]
predicate flatten(loc x, loc output){
| x == 0 => {output :-> 0}
| x != 0 => {[x, 2] ** x :-> x0 ** sll(x0) ** x + 1 :-> xnxt ** temp outputnxt
  ** flatten(xnxt, outputnxt) ** func append(x, outputnxt, output)} }
\end{lstlisting}

\subsection{Avoiding Excessive Heap Manipulation with Read-Only Locations}

The existing \suslik depends on the set theory to express the pure
relation. However, it is not trivial to automatically generate the
pure part of SSL specifications from the functional specifications. To
see why the set theory is needed, the following simple example shows
the functionality of the set theory, with \lstinline{sll_n} being the
\textbf{s}ingle-\textbf{l}inked \textbf{l}ist with \textbf{n}o set.

\begin{lstlisting}[language=SynLang]
predicate sll_n(loc x) {
|  x == 0        => {true; emp }
|  not (x == 0)  => { [x, 2] ** x :-> v ** (x + 1) :-> xnxt ** sll(xnxt) } }
predicate copy(loc x, loc y) {
|  x == 0        => {y == 0; emp }
|  not (x == 0)  => { [y, 2] ** y :-> v ** (y + 1) :-> ynxt ** [x, 2] ** x :-> v
     ** (x + 1) :-> xnxt ** copy(xnxt, ynxt) } }
\end{lstlisting}

While the intent of the function predicate \lstinline{copy} is to copy
the list \lstinline{x} to \lstinline{y}, without the set theory, the
output program will be somewhat surprising to see:

\begin{lstlisting}[language=SynLang]
{sll_n(x)}
void copy (loc x, loc y) {
  if (x == 0) {
  } else {
    let n = *(x + 1); copy(n, y); let y01 = *y; let y0 = malloc(2); *y = y0;
    *(y0 + 1) = y01; let vy = *y0 *x = vy; } }
{copy(x, y)}
\end{lstlisting}

The problem here is that, when we have the pure relation in the
predicate to indicate that the values are the same, the synthesiser
finds another possible way: instead of copying the value of
\lstinline{x} to \lstinline{y}, we can just change the value of x to
initial value \lstinline{vy} after \lstinline[language = c]{malloc}.
This is not the user intent, and the output program is not correct.
Turns out, the solution is not that difficult: we simply need add a
new kind of heaplet in the specification language, call
\textit{constant points-to}, which has a similar idea as read-only
borrows~\cite{costea2020concise}.
%
%
The only difference of the constant points-to from the original
\textit{points-to} heaplet is that the value of the location is
constant, which means that the \writer rule in SSL is not applicable. By this way, the extended \suslik will not consider the modification of the input location, thus provides the correctness mechanism (in \autoref{sec:soundness}) for the translation of Pika.

\section{Evaluation}
\label{sec:evaluation}

In this section, we evaluate \tool's expressiveness. A secondary objective is to
evaluate \tool's performance. The performance evaluation is done largely to put
\tool into context by comparing it to a prominent functional programming
language (Haskell). The main purpose of \tool is to increase the expressiveness
of SuSLik, which is the reason for the primary evaluation objective.
Towards these goals, we answer the following research questions:

\begin{itemize}
\item \RQ{1}: Is the performance of the synthesised code competitive
  with code generated from traditional functional language compilers?

\item \RQ{2}: In concrete terms, how does \theFnLang's expressivity
  compare with the expressivity of \suslik specifications for programs
  written in a functional style?

  \item \RQ{3}: What are the failure modes of our approach?
\end{itemize}

Our implementation and benchmarks are available in supplementary material. The experiments were conducted on a 2021 MacBook Pro with an M1 processor and 32 GB of RAM. We used GHC version 9.8.1 and Apple Clang version 13.0.0.

For \RQ{1}, we run benchmarks and compare execution time. The benchmarks we selected are a series of functions that manipulate data structures such as lists and trees, which covers different common abstract operations (like \code{map}, \code{filter}, \code{fold}). 
The comparison is between Haskell functions based on user-defined data
structures and C functions generated from \tool's specifications (via
the extended \suslik), parameterised with the data structures of large size. We recorded both execution time with and
without optimisation of compilers.


The results are shown in \autoref{fig:performance}. Our findings are as follows:

\begin{itemize}
\item When compiled to an executable run without optimisations, the C
  programs generated by \tool are faster than Haskell programs. And we can also observe that the
  speed difference is larger for functions with more complex data
  structures.
  \item  When compiled to an executable with optimisations:

  \begin{itemize}
    \item For functions with complex data structures, the comparison
      is similar to the case without optimisation.

    \item For functions for the singly-linked list GHC's optimisation
      is very powerful, resulting in much better performance
      than C programs generated by \tool. 
      With more tests, we found out the performance after GHC's
      optimisation is similar to the one using Haskell's built-in list.
      We believe this is because GHC's optimisations are fine-tuned to
      optimise code manipulating list-like data structures and
      use the same optimisation as for Haskell's built-in list. That
      said, similar observations regarding GHC do not hold on other
      complex data structures, such as trees.
  \end{itemize}
\end{itemize}

\begin{figure}[t]
  \begin{center}
    \begin{subfigure}[b]{0.49\textwidth}
      \centering
      \includegraphics[width=\textwidth]{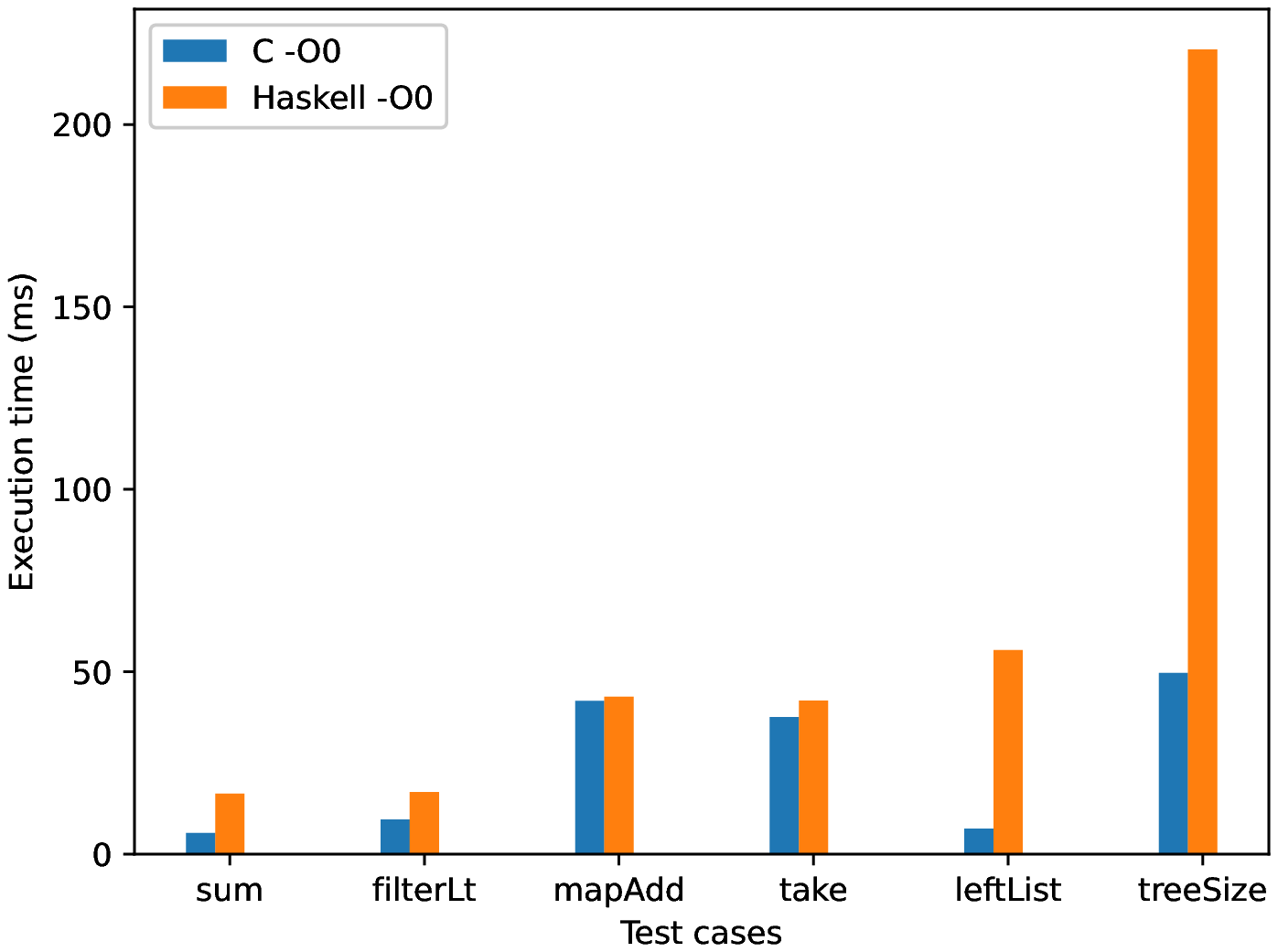}
      \caption{Benchmarks without optimizations}
    \end{subfigure}
    \begin{subfigure}[b]{0.49\textwidth}
      \centering
      \includegraphics[width=\textwidth]{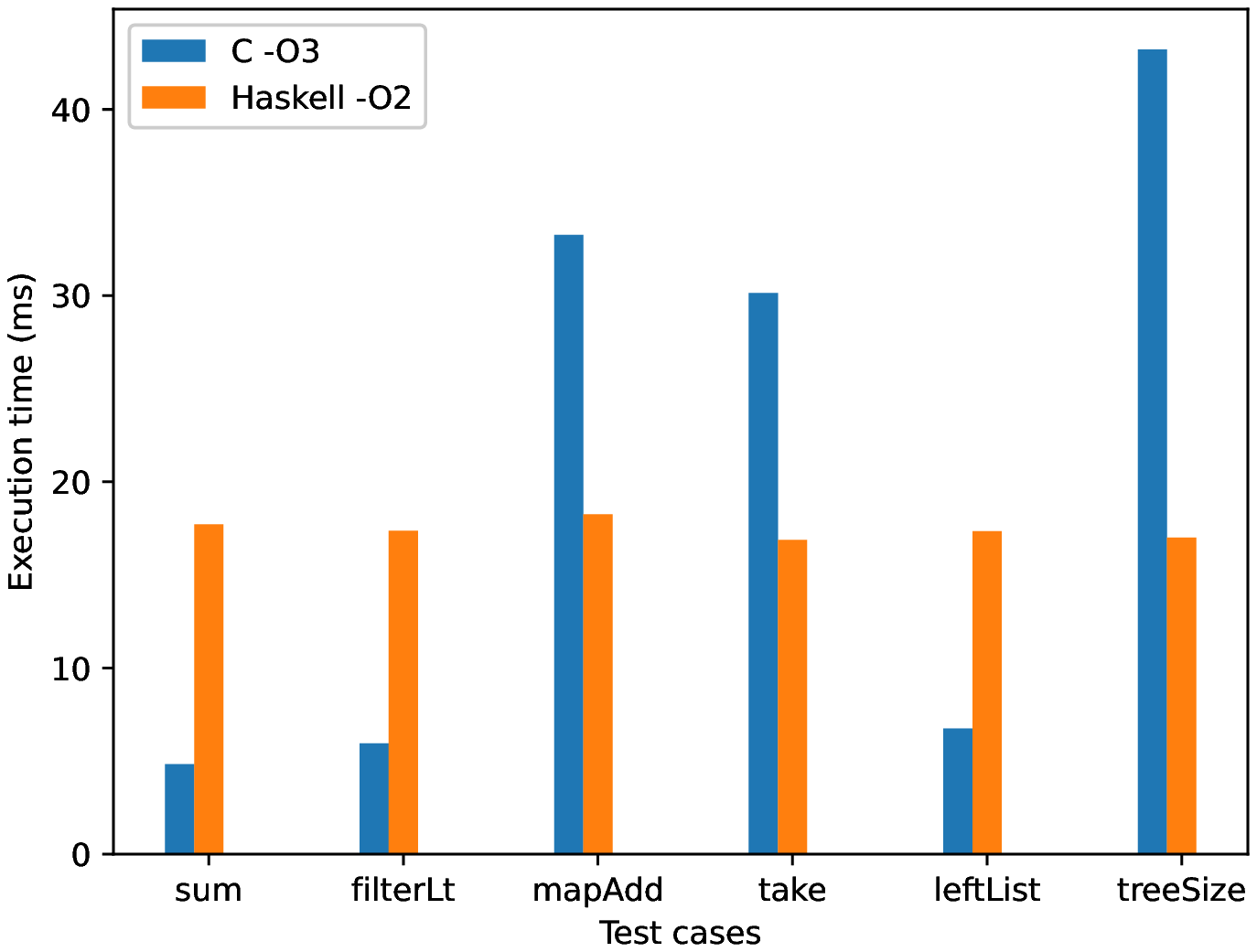}
      \caption{Benchmarks with optimizations}
    \end{subfigure}
  \end{center}
  \caption{Performance of C functions generated by \tool compared to
    Haskell/GHC}
  \label{fig:performance}
\end{figure}

\begin{figure}
\begin{center}
  \begin{tabular}{l@{\ \ }c@{\ \ }c@{\ \ }c@{\ \ }c@{\ \ }c@{\ \ }c}
\toprule 
\# & Task Name & \tool AST & \suslik AST & $\frac{|\text{\tool
                                    AST}|}{|\text{\suslik AST}|}$ &
                                                                 Compile Time & Synthesis Time
\\
\midrule
1 & \verb|cons| & 76 & 123 & 0.618 & 0.012 & 5.134\\
2 & \verb|plus| & 46 & 91 & 0.505 & 0.002 & 5.992\\
3 & \verb|add1Head| & 64 & 109 & 0.587 & 0.006 & 5.114\\
4 & \verb|listId| & 62 & 107 & 0.579 & 0.005 & 4.906\\
5 & \verb|add1HeadDLL| & 74 & 146 & 0.507 & 0.007 & 10.618\\
6 & \verb|even| & 19 & 42 & 0.452 & 0.001 & 3.999\\
7 & \verb|foldr| & 63 & 113 & 0.558 & 0.005 & 5.454\\
8 & \verb|sum| & 58 & 103 & 0.563 & 0.004 & 5.125\\
9 & \verb|filterLt| & 84 & 147 & 0.571 & 0.009 & 6.104\\
10 & \verb|mapAdd| & 71 & 116 & 0.612 & 0.007 & 5.031\\
11 & \verb|leftList| & 116 & 156 & 0.744 & 0.008 & 7.722\\
12 & \verb|treeSize| & 78 & 126 & 0.619 & 0.007 & 5.980\\
13 & \verb|take| & 119 & 212 & 0.561 & 0.012 & 10.447\\
\bottomrule
\end{tabular}

\end{center}
  \caption{Statistics on the benchmarks: \tool spec size v, generated SSL spec size, translator
    performance, and synthesiser performance. All times are in seconds.}
  \label{fig:size-comparison}
\end{figure}


To answer \RQ{2}, we first find some common patterns for \tool programs in the benchmarks shown in \autoref{fig:size-comparison}: (1) pattern matching on ADTs (\#9, 10, etc.), (2) code reusability (\#3 vs 5, \#7 vs 8). Those features are not directly expressible in \suslik because of the low-level nature of SSL. For example, the \verb|add1Head| and \verb|add1HeadDLL| functions share the same function definition, where the only difference is the type layout used; but the \suslik specifications need to be treated separately, which makes the codes more complex.
To make some objective observations on the expressivity, we measure the number
of nodes in the input \tool AST and find it consistently fewer than the number of AST
nodes in the generated \suslik specification.

To address \RQ{3}, note that a particular failure mode occurs when
\tool source code reuses a variable in a way that the constraints
violate separation logic. For example, see the \verb|selfAppend|
example in \autoref{fig:fail} which uses its argument twice. 
We could have addressed this issue by introducing a lightweight linear
type system to \tool, but have not carried out this exercise yet.
Another kind of failure occurs when \suslik fails to synthesise an
implementation of the generated specification: handling these
failures is beyond the scope of this work.

\begin{figure}[t]
\setlength{\abovecaptionskip}{5pt}
\begin{center}
\begin{lstlisting}
selfAppend : List -> List;
selfAppend xs := instantiate [Sll, Sll] Sll append xs xs;

append : List -> List -> List;
append (Nil) ys := ys;
append (Cons x xs) ys := instantiate [Int, Sll[mutable]] Sll cons
                                       (addr x) (append xs ys);
\end{lstlisting}
\end{center}
\caption{A \tool specification not supported by \suslik.}
  \label{fig:fail}
\end{figure}


\section{Discussion}

We have given a translation from a high-level functional language into
SSL specifications to be given to a program synthesiser.

In doing so, we have revealed a close connection between, on one hand,
algebraic data types and recursive pattern-matching functions and, on
the other hand, SSL inductive predicates. The soundness of this
connection is demonstrated by \autoref{thm:gen-soundness} in
\autoref{sec:formal}.

Beyond the theory, this connection can be exploited in three directions:

\begin{itemize}
  \item Increased type safety: Algebraic data types allow you to distinguish between types that have the same
    runtime heap representation.

  \item More reusability: An algebraic data type can have multiple layouts, each of which gives a different possible
    runtime heap representation for the ADT. As \tool functions are defined only in terms of algebraic data types, they naturally get the \emph{polymorphism} of the ADT by being able to work with any layout of the ADT. 

  \item Greater succinctness: When working at this higher level of abstraction, it generally takes less code
    as you are not frequently manipulating heap locations. The only mention of heap locations is in reusable layout
    definitions. This can also give greater clarity.
\end{itemize}


\section{Related Work}

\tool is built upon the \suslik synthesis framework. \suslik provides
a synthesis mechanism for heap-manipulating programs using a variant
of Separation Logic~\cite{polikarpova:2019:suslik}. However, it does
not have any high-level abstractions. In particular, writing \suslik
specifications directly involves a significant amount of pointer
manipulation. Further, it does not provide abstraction over specific
memory layouts. As described in \autoref{sec:language}, \tool
addresses these limitations.

\tname{Dargent} language~\cite{chen:2023:dargent} also includes a
notion of layouts and layout polymorphism for a class of algebraic
data types, which differs from our treatment of layouts in two primary
ways:

\begin{enumerate}

\item In \tool, abstract memory locations (with offsets) are used. In
  contrast, \tname{Dargent} uses offsets that are all
  relative to a single ``beginning'' memory location. The \tool
  approach is more amenable to heap allocation, though
  this requires a separate memory manager of some kind. This is
  exposed in the generated language with \verb|malloc| and
  \verb|free|.
  On the other hand, the technique taken by \tname{Dargent} allows
  for greater control over memory management. This makes
  dealing with memory more complex for the programmer, but it is no
  longer necessary to have a separate memory manager.

  \item Algebraic data types in the present language include
    \emph{recursive} types and,
    as a result, \tool has recursive layouts for these ADTs. This
    feature is not currently available in \tname{Dargent}.
  \end{enumerate}

Furthermore, layout polymorphism also works differently. While
\tname{Dargent} tracks layout instantiations at a type-level with type
variables, in the present work we simply only check to see if a layout
is valid for a given type when type-checking. In particular, we cannot
write type signatures that \textit{require} the same layout in
multiple parts of the type (for instance, in a function type
\verb|List -> List| we have no way at the type-level of requiring that
the argument \verb|List| layout and the result \verb|List| layout are
the same). This more rudimentary approach that \tool currently takes
could be extended in future work.
Overall, the examples in the \tname{Dargent} paper tend to focus on
the manipulation of integer values. In contrast, we have focused
largely on data structure manipulation, which follow the primary
motivation of \suslik.

\tname{Synquid} is another synthesis framework with a functional
surface language. While \tname{Synquid} allows an even higher-level
program specification than \tool through its liquid types, it does not
provide any access to low-level data structure
representation.~\cite{polikarpova:2016:synquid} In contrast, \tool's
level of abstraction is similar to that of a traditional functional
language but, similar to \tname{Dargent}, it also allows control over
the data structure representation in memory.


\section{Future Work and Conclusion}

We make the following observations based on our experience of
developing and using \tool.


By allowing layouts to use multiple SSL parameters, we would be able
to give a greater variety of layouts associated with an ADT. For
instance, the \verb|List| data type used in the examples could have a
doubly-linked list layout in addition to the singly-linked list layout
\verb|Sll|. Note that any existing \tool function defined over
\verb|List| will continue to work with no modification with these new
layouts. Defunctionalisation and lambda lifting can also be used to
implement true higher-order functions.

It is possible to do inference of some layouts, for example in
\verb|mapAdd1| we would usually want to use the same layout as the
argument layout, but we leave this for future work. Another approach
is to introduce type variables that correspond to layouts, as done in
the series of works on the \tname{Dargent}
tool~\cite{chen:2023:dargent}. We leave this approach for future work as well.

Reverse transformation deserves further investigation: if we go from
an SSL specification to \tool program and then compile to, \eg, C, can
we synthesise additional programs that traditional SSL synthesisers
would struggle with? What are the limitations of this approach?



We may be able to expose more of the synthesis mechanism in the \tool language. For example, generate an SSL specification given only a \tool type
signature (and corresponding \generate{} directive). This could
combine well with additional polymorphism, as we could utilise the
free theorems that are given by a polymorphic type signature to further
constrain the resulting specification.

Finally, is it possible to derive translations for languages such as \tool
from abstract machine semantics? In this paper, we have given a
language with abstract machine semantics. We then give a translation
of that language into SSL. We then show that the final states given by
the abstract machine semantics are models for the SSL propositions
produced by our translation. But is it possible to begin by specifying the
abstract machine semantics and then mathematically (or automatically)
\textit{derive} an appropriate translation into SSL, with the
requirement that the translation satisfies the soundness theorem?

In conclusion, we have presented \tool: a high-level functional
specification language that paves the way for the efficient synthesis of a
verifiably correct imperative code with in-place memory updates that
is comparable in efficiency to the handwritten C.

\bibliography{alex}

\appendix

\section{Implementation}

The translation occurs as a sequence of several stages:
\begin{enumerate}
  \item \stage{Type checking and elaboration}. This stage does type-checking. It also
  annotates the abstract syntax tree with the fresh names that will be used in the resulting SuSLik code. The
  number of SuSLik names needed for a given layout is determined during type checking, which is why this
  annotation is performed at this stage. \verb|instantiate|s are also instantiated for recursive calls during
  this stage.

  \item \stage{Unfold empty constructors}. Constructors that are assigned \verb|emp| by a layout need to be
  handled specially. This is because an \verb|== 0| constraint must be generated at the end, as opposed to just a \verb|emp|
  heaplet.

  \item \stage{Unfold pattern matches using layouts}. This stage partially translates a pattern match into
  its SuSLik equivalent. The AST is annotated with the heaplets corresponding to the layout branch for the
  particular constructor in the pattern match.
  The other part of the pattern match translation, the generation of the corresponding Boolean conditionals, is done during the final stage.

  \item \stage{Insert copying predicate applications}.

  \item \stage{Translate \texttt{let}s}. Each \verb|let| is turned into a SuSLik equality constraint.

  \item \stage{Unfold constructor applications}. In the expressions that constitute the function body,
  translate the constructor applications into heaplets using the specified layouts. Function calls are translated
  into the corresponding inductive predicate applications. For base type expressions (that is, of type \verb|Int| or \verb|Bool|),
  pure SuSLik constraints are generated. During this phase, \verb|if-then-else| expressions are translated into SuSLik equality constraints
  with \verb|... ? ... : ...| on the right-hand side of the equalities.

  \item \stage{Generation}. The Boolean conditions corresponding to each pattern match are generated and the final inductive predicate
  is generated.
\end{enumerate}
In Appendix~\ref{sec:ImplEx} the \verb|filterLt9| example is used to go through each stage, step-by-step.


\section{Algebraic Data Type and Layout Definitions For Examples}
\label{sec:example-defs}

\begin{lstlisting}
data List := Nil | Cons Int List;
data Tree := Leaf | Node Int Tree Tree;
data ListOfLists := LNil | LCons List ListOfLists;
data Zipped := ZNil | ZCons Int Int Zipped;

Sll : List >-> layout[x];
Sll (Nil) := emp;
Sll (Cons head tail) := x :-> head, (x+1) :-> tail, Sll tail;

TreeLayout : Tree >-> layout[x];
TreeLayout (Leaf) := emp;
TreeLayout (Node payload left right) :=
  x :-> payload,
  (x+1) :-> left,
  (x+2) :-> right,
  TreeLayout left,
  TreeLayout right;

ListOfListsLayout : ListOfLists >-> layout[x];
ListOfListsLayout (LNil) := emp;
ListOfListsLayout (LCons head tail) :=
  x :-> head, (x+1) :-> tail,
  ListOfListsLayout tail,
  Sll head;

ZippedLayout : Zipped >-> layout[x];
ZippedLayout (ZNil) := emp;
ZippedLayout (ZCons fst snd rest) :=
  x :-> fst,
  (x+1) :-> snd,
  (x+2) :-> rest,
  ZippedLayout rest;
\end{lstlisting}

\section{Reference Implementation for Examples}
\label{sec:ref-impls}

\subsection{filterLt9}

\fsImpl

\begin{lstlisting}
%generate filterLt9 [Sll] Sll

filterLt9 : List -> List;
filterLt9 (Nil) := Nil;
filterLt9 (Cons head tail)
  | head < 9       := filterLt9 tail;
  | not (head < 9) := Cons head (filterLt9 tail);
\end{lstlisting}

\genSuSLik

\begin{lstlisting}
predicate filterLt9__rw_Sll__ro_Sll(loc __p_x0, loc __r_x) {
| (__p_x0 == 0) => { __r_x == 0 ; emp }
| ((not (__p_x0 == 0)) && (head < 9)) => {
    __p_x0 :-> head ** (__p_x0+1) :-> tail ** [__p_x0,2] **
    filterLt9__rw_Sll__ro_Sll(tail, __r_x) }
| ((not (__p_x0 == 0)) && (not (head < 9))) => {
    __p_x0 :-> head ** (__p_x0+1) :-> tail ** [__p_x0,2] **
    filterLt9__rw_Sll__ro_Sll(tail, __p_x3) **
    __r_x :-> head ** (__r_x+1) :-> __p_x3 ** [__r_x,2] }
}
\end{lstlisting}

\refSuSLik

\begin{lstlisting}
predicate filter_base2(loc x, loc y){
| y == 0 => {x == 0 ; emp}
| (not (y == 0)) && (vy < 9)   => {
    [y,2] ** y :-> vy ** (y + 1) :-> ynxt **
    filter_base2(x, ynxt)}
| (not (y == 0)) && (not (vy < 9))   => {
     [y, 2] ** y :-> vy ** (y + 1) :-> ynxt **
     [x, 2] ** x :-> vy ** (x + 1) :-> xnxt **
     filter_base2(xnxt, ynxt)}
}
\end{lstlisting}

\subsection{fold}

\fsImpl

\begin{lstlisting}
%generate fold_List [Int, Sll] (Ptr Int)

fold_List : Int -> List -> Ptr Int;
fold_List z (Nil) := z;
fold_List z (Cons x xs) :=
  instantiate
    [Ptr Int, Ptr Int]
    (Ptr Int)
    f
    (addr x)
    (fold_List z xs);
\end{lstlisting}

\genSuSLik

\begin{lstlisting}
predicate fold_List__Ptr_Int__Int__ro_Sll(int __p_0, loc __p_x1
                                         ,loc __r) {
| (__p_x1 == 0) => { __r :-> __p_0 }
| (not (__p_x1 == 0)) => {
  __p_x1 :-> x ** (__p_x1+1) :-> xs ** [__p_x1,2] **
  func f__Ptr_Int__Ptr_Int__Ptr_Int(__p_x1, __p_2, __r) **
  fold_List__Ptr_Int__Int__ro_Sll(__p_0, xs, __p_2)
  ** temploc __p_2 }
}
\end{lstlisting}

\refSuSLik

\begin{lstlisting}
predicate fold_sll(int init, loc x, loc ret){
|  x == 0 => { func sll_copy(init, ret)}
|  not (x == 0) => {
    [x, 2] ** x :-> v ** (x + 1) :-> xnxt ** temploc t **
    func f(x, t, ret) **
    fold_sll(init, xnxt, t)
}
}
\end{lstlisting}

\subsection{maximum}

\fsImpl

\begin{lstlisting}
%generate maximum [Sll] Int

maximum : List -> Int;
maximum (Nil) := 0;
maximum (Cons x xs) :=
  let i := maximum xs
  in
  if i < x
    then x
    else i;
\end{lstlisting}

\genSuSLik

\begin{lstlisting}
predicate maximum__Int__ro_Sll(loc __p_x0, int __r) {
| (__p_x0 == 0) => { __r == 0 ; emp }
| (not (__p_x0 == 0)) => {
      i == __p_1 && __r == ((i < x) ? x : i) && __temp_0 == __p_1
        ;
      __p_x0 :-> x ** (__p_x0+1) :-> xs ** [__p_x0,2] **
      maximum__Int__ro_Sll(xs, __temp_0) }
}
\end{lstlisting}

\refSuSLik

\begin{lstlisting}
predicate maximum(loc x, int ret) {
| x == null => {ret == 0; emp}
| not (x == null) => {
    ret == (ret0 <= v ? v : ret0);
    [x, 2] ** x :-> v ** x + 1 :-> xnxt **
    maximum(xnxt, ret0)
}
}
\end{lstlisting}

\section{More Examples}

\subsection{cons}

This example prepends an item to a list. In particular, we have the function instantiated with the singly-linked list layout, so
it will synthesize a function to perform this task on singly-linked lists.

\fsImpl
\begin{lstlisting}
%generate cons [Int, Sll] Sll

cons : Int -> List -> List;
cons x xs := Cons x xs;
\end{lstlisting}

\genSuSLik
\begin{lstlisting}
predicate cons__rw_Sll__Int__ro_Sll(int __p_0, loc __p_x1, loc __r_x)
{
| true => { __r_x :-> __p_0 ** (__r_x+1) :-> __p_x1 ** [__r_x,2] }
}
\end{lstlisting}

\refSuSLik
\begin{lstlisting}
predicate cons(int v, loc x, loc ret) {
| true => {[ret, 2] ** ret :-> v ** func cp(x, ret) + 1 ** sll_c(x)}
}
\end{lstlisting}

\subsection{car}
This example would get the head of a list. The synthesizer is currently synthesizing a recursive function which does
not have the correct behavior.

\fsImpl

\begin{lstlisting}
%generate car [Sll] Int

car : List -> Int;
car (Nil) := 0;
car (Cons x xs) := x;
\end{lstlisting}

\genSuSLik

\begin{lstlisting}
predicate car__Int__ro_Sll(loc __p_x0, int __r) {
| (__p_x0 == 0) => { __r == 0 ; emp }
| (not (__p_x0 == 0)) => {
    __r == x ; __p_x0 :-> x ** (__p_x0+1) :-> xs ** [__p_x0,2] }
}
\end{lstlisting}

\subsection{singleton}

A function which creates a singleton list out of the given integer argument.

\fsImpl

\begin{lstlisting}
%generate singleton [Int] Sll

singleton : Int -> List;
singleton x := Cons x (Nil);
\end{lstlisting}

\genSuSLik

\begin{lstlisting}
predicate singleton__rw_Sll__Int(int __p_0, loc __r_x) {
| true => { __r_x :-> __p_0 ** (__r_x+1) :-> 0 ** [__r_x,2] }
}
\end{lstlisting}

\refSuSLik

\begin{lstlisting}
predicate singleton_list(int n, loc x) {
| true => { [x, 2] ** x :-> n ** (x+1) :-> 0 }
}
\end{lstlisting}

\subsection{map}

\fsImpl

\begin{lstlisting}
%generate map [Sll] Sll

map : List -> List;
map (Nil) := Nil;
map (Cons x xs) := Cons (instantiate [Int] Int f x) (map xs);
\end{lstlisting}

\genSuSLik

\begin{lstlisting}
predicate map__rw_Sll__ro_Sll(loc __p_x0, loc __r_x) {
| (__p_x0 == 0) => { __r_x == 0 ; emp }
| (not (__p_x0 == 0)) => {
    __r_x == __p_1 ; __p_x0 :-> x ** (__p_x0+1) :-> xs **
    [__p_x0,2] ** func f__Int__Int(x, __p_1) **
    map__rw_Sll__ro_Sll(xs, __p_x3) **
    (__r_x+1) :-> __p_x3 ** [__r_x,2] }
}
\end{lstlisting}

\refSuSLik

\begin{lstlisting}
predicate sll_map(loc x, loc y) {
|  x == 0        => { y == 0 ; emp}
|  not (x == 0)  => {
     [x, 2] ** x :-> h ** (x + 1) :-> xnxt ** [y, 2] **
     func f(x, y) ** (y + 1) :-> ynxt **
     sll_map(xnxt, ynxt) }
}
\end{lstlisting}

\subsection{snoc}

This example prepends an item to a singly-linked list.

\fsImpl

\begin{lstlisting}
%generate snoc [Sll, Int] Sll

snoc : List -> Int -> List;
snoc (Nil) i := Cons i (Nil);
snoc (Cons x xs) i := Cons x (snoc xs i);
\end{lstlisting}

\genSuSLik

\begin{lstlisting}
predicate snoc__rw_Sll__ro_Sll__Int(loc __p_x0, int __p_1, loc __r_x)
{
| (__p_x0 == 0) => { __r_x :-> __p_1 ** (__r_x+1) :-> 0 **
                     [__r_x,2] }
| (not (__p_x0 == 0)) => {
    __p_x0 :-> x ** (__p_x0+1) :-> xs ** [__p_x0,2] **
    snoc__rw_Sll__ro_Sll__Int(xs, __p_1, __p_x3) **
    __r_x :-> x ** (__r_x+1) :-> __p_x3 ** [__r_x,2] }
}
\end{lstlisting}

\subsection{reverse}

This example reverses a singly-linked list.

\fsImpl

\begin{lstlisting}
%generate reverse [Sll] Sll

reverse : List -> List;
reverse (Nil) := Nil;
reverse (Cons x xs) :=
  instantiate [Sll[mutable], Int] Sll snoc (reverse xs) x;
\end{lstlisting}

\genSuSLik

\begin{lstlisting}
predicate reverse__rw_Sll__ro_Sll(loc __p_x0, loc __r_x) {
| (__p_x0 == 0) => { __r_x == 0 ; emp }
| (not (__p_x0 == 0)) => {
    __p_x0 :-> x ** (__p_x0+1) :-> xs ** [__p_x0,2] **
    func snoc__rw_Sll__rw_Sll__Int(__p_x1, x, __r_x) **
    reverse__rw_Sll__ro_Sll(xs, __p_x1) **
    temploc __p_x1 }
}
\end{lstlisting}

\refSuSLik

\begin{lstlisting}
predicate reverse(int init, loc x, loc ret){
|  x == 0 => { ret :-> init}
|  not (x == 0) => {
    [x, 2] ** x :-> v ** (x + 1) :-> xnxt ** temploc t **
    reverse(init, xnxt, t) ** func tail(t, x, ret)
}
}

predicate tail(loc x, int v, loc ret){
| x == null => {[ret, 2] ** ret :-> v ** ret + 1 :-> 0}
| not (x == null) => {
    [x, 2] ** x :-> v ** x+1 :-> xnxt **
    [ret, 2] ** ret :-> v ** ret+1 :-> retnxt **
    tail(xnxt, v, retnxt)
}
}
\end{lstlisting}

\subsection{append}

This example appends two singly-linked lists together.

\fsImpl

\begin{lstlisting}
%generate append [Sll, Sll] Sll

append : List -> List -> List;
append (Nil) ys := ys;
append (Cons x xs) ys :=
  instantiate [Ptr Int, Sll[mutable]] Sll
    cons
    (addr x)
    (append xs ys);
\end{lstlisting}

\genSuSLik

\begin{lstlisting}
predicate append__rw_Sll__ro_Sll__ro_Sll(loc __p_x0, loc __p_x1
                                        ,loc __r_x) {
| (__p_x0 == 0) => { func Sll__copy(__p_x1, __r_x) }
| (not (__p_x0 == 0)) => {
     __p_x0 :-> x ** (__p_x0+1) :-> xs ** [__p_x0,2] **
     func cons__rw_Sll__Ptr_Int__rw_Sll(__p_x0, __p_x2, __r_x) **
     append__rw_Sll__ro_Sll__ro_Sll(xs, __p_x1, __p_x2) **
     temploc __p_x2 }
}
\end{lstlisting}

\refSuSLik

\begin{lstlisting}
predicate append(loc x1, loc x2, loc ret){
| x1 == 0 => {sll_c(x2) ** func sll_copy(x2, ret)}
| not (x1 == 0) => {
    [x1, 2] ** x1 :-> v ** x1 + 1 :-> x1nxt **
    temploc t ** append(x1nxt, x2, t) **
    func cons(x1, t, ret)

}
}
\end{lstlisting}

\subsection{zip}

This example takes two lists and zips them into a list of pairs.

\fsImpl

\begin{lstlisting}
%generate zip [Sll, Sll] ZippedLayout

zip : List -> Zipped;
zip (Nil) ys := ZNil;
zip xs (Nil) := ZNil;
zip (Cons x xs) (Cons y ys) := ZCons x y (zip xs ys);
\end{lstlisting}

\genSuSLik

\begin{lstlisting}
predicate zip__rw_ZippedLayout__ro_Sll__ro_Sll(loc __p_x0, loc __p_x1
                                              ,loc __r_x) {
| (__p_x0 == 0) => { __r_x == 0 ; emp }
| (__p_x1 == 0) => { __r_x == 0 ; emp }
| ((not (__p_x0 == 0)) && (not (__p_x1 == 0))) => {
     __p_x0 :-> x ** (__p_x0+1) :-> xs ** [__p_x0,2] **
     __p_x1 :-> y ** (__p_x1+1) :-> ys ** [__p_x1,2] **
     zip__rw_ZippedLayout__ro_Sll__ro_Sll(xs, ys, __p_x2) **
     __r_x :-> x ** (__r_x+1) :-> y ** (__r_x+2) :-> __p_x2 **
     [__r_x,3] }
}
\end{lstlisting}

\subsection{zipWith}

This example is similar to \verb|zip|, but instead of pairing it applies a binary
function pairwise to the items from the two lists.

\fsImpl

\begin{lstlisting}
%generate zipWith [Sll, Sll] Sll

zipWith : List -> List -> List;
zipWith (Nil) (Nil) := Nil;
zipWith (Nil) (Cons y ys) := Nil;
zipWith (Cons x xs) (Nil) := Nil;
zipWith (Cons x xs) (Cons y ys) :=
  Cons
    (instantiate [Int, Int] Int f x y)
    (zipWith xs ys);
\end{lstlisting}

\genSuSLik

\begin{lstlisting}
predicate zipWith__rw_Sll__ro_Sll__ro_Sll(loc __p_x0, loc __p_x1
                                         ,loc __r_x) {
| ((__p_x0 == 0) && (__p_x1 == 0)) => { __r_x == 0 ; emp }
| ((__p_x0 == 0) && (not (__p_x1 == 0))) => { __r_x == 0
                                              ; ro_Sll(__p_x1) }
| ((__p_x1 == 0) && (not (__p_x0 == 0))) => { __r_x == 0
                                              ; ro_Sll(__p_x0) }
| ((not (__p_x0 == 0)) && (not (__p_x1 == 0))) => {
   __r_x == __p_2
     ;
   __p_x0 :-> x ** (__p_x0+1) :-> xs ** [__p_x0,2] **
   __p_x1 :-> y ** (__p_x1+1) :-> ys ** [__p_x1,2] **
   func f__Int__Int__Int(x, y, __p_2) **
   zipWith__rw_Sll__ro_Sll__ro_Sll(xs, ys, __p_x5) **
   (__r_x+1) :-> __p_x5 ** [__r_x,2] }
}
\end{lstlisting}

\refSuSLik

\begin{lstlisting}
predicate zip_withf(loc x, loc y, loc r) {
| x == 0 && y == 0 => { r == null ; emp }
| y == 0 && not (x == 0) => { r == null ; sll_c(x) }
| x == 0 && not (y == 0) => { r == null ; sll_c(y) }
| (not (x == 0)) && (not (y == 0)) => {
    [x, 2] ** x :-> a ** (x+1) :-> nxtX **
    [y, 2] ** y :-> b ** (y+1) :-> nxtY **
    [r, 2] ** func f(x, y, r) ** (r+1) :-> nxtR **
    zip_withf(nxtX, nxtY, nxtR)
  }
}
\end{lstlisting}

\subsection{scanr}

This is similar to a right fold, but it collects all of the intermediate results
in the result list.

\fsImpl

\begin{lstlisting}
%generate scanr[Int, Sll] Sll

scanr : Int -> List -> List;
scanr z (Nil) := Cons z (Nil);
scanr z (Cons x xs) :=
  let xs2 := scanr z xs
  in
  let x2 := instantiate [Int, Int] Int f x
              (instantiate [Sll[mutable]] Int head xs2)
  in
  Cons x2 xs2;
\end{lstlisting}

\genSuSLik

\begin{lstlisting}
predicate scanr__rw_Sll__Int__ro_Sll(int __p_0, loc __p_x1
                                    ,loc __r_x) {
| (__p_x1 == 0) => { __r_x :-> __p_0 ** (__r_x+1) :-> 0 **
                     [__r_x,2] }
| (not (__p_x1 == 0)) => {
    xs2 == __p_x3 && x2 == __p_7
      ;
    __p_x1 :-> x ** (__p_x1+1) :-> xs ** [__p_x1,2] **
    scanr__rw_Sll__Int__ro_Sll(__p_0, xs, __p_x3) **
    func f__Int__Int__Int(x, __p_8, __p_7) **
    func head__Int__rw_Sll(__p_x3, __p_8) **
    __r_x :-> x2 ** (__r_x+1) :-> __p_x3 ** [__r_x,2] }
}
\end{lstlisting}

\refSuSLik

\begin{lstlisting}
predicate scanr_sum(int init, int tmp, loc x, loc ret){
|  x == 0 => {init == tmp ; [ret, 2] **  ret :-> init **
                            ret + 1 :-> 0}
|  not (x == 0) => { tmp == v + tmp0;
    [x, 2] ** x :-> v ** (x + 1) :-> xnxt **
    [ret, 2] ** ret :-> tmp ** ret + 1 :-> retnxt **
    scanr_sum(init, tmp0, xnxt, retnxt)
}
}
\end{lstlisting}

\subsection{selfAppend}

This example would append a list to itself. Due to the fact that there is currently
no exposed copy operation, the fact that this function refers to its argument multiple
times prevents it from synthesizing.

\fsImpl

\begin{lstlisting}
%generate selfAppend[Sll] Sll

selfAppend : List -> List;
selfAppend xs :=
  instantiate [Sll, Sll] Sll append xs xs;
\end{lstlisting}

\genSuSLik

\begin{lstlisting}
predicate selfAppend__rw_Sll__ro_Sll(loc __p_x0, loc __r_x) {
| true => {
    func append__rw_Sll__ro_Sll__ro_Sll(__p_x0, __p_x0, __r_x) }
}
\end{lstlisting}

\subsection{take}

This would return the first \verb|n| items of the given list. In the higher-level implementation, we 
use a machine \verb|Int| for the count argument. As SuSLik does not currently have the capability to
reason inductively about machine \verb|int|s, this does not synthesize. This limitation could possibly be
addressed by using an ADT unary representation for natural numbers.

\fsImpl

\begin{lstlisting}
%generate take [Int, Sll] Sll

take : Int -> List -> List;
take i (Nil) := Nil;
take i (Cons head tail)
  | i == 0 := Nil;
  | not (i == 0) := Cons head (take (i - 1) tail);
\end{lstlisting}

\genSuSLik

\begin{lstlisting}
predicate take__rw_Sll__Int__ro_Sll(int __p_0, loc __p_x1, loc __r_x)
{
| (__p_x1 == 0) => { __r_x == 0 ; emp }
| ((not (__p_x1 == 0)) && (__p_0 == 0)) => { __r_x == 0
                                               ; ro_Sll(__p_x1) }
| ((not (__p_x1 == 0)) && (not (__p_0 == 0))) => {
     __p_x1 :-> head ** (__p_x1+1) :-> tail ** [__p_x1,2] **
     take__rw_Sll__Int__ro_Sll((__p_0 - 1), tail, __p_x2) **
     __r_x :-> head ** (__r_x+1) :-> __p_x2 ** [__r_x,2] }
}
\end{lstlisting}

\refSuSLik

\begin{lstlisting}
predicate take(loc x, int n, int ret){
| true => {
    temploc t ** func number(n, t) ** func take_helper(x, t, ret)
}
}


predicate take_helper(loc x, loc helper, int ret){
| x == 0 => { ret == 0; sll_c(helper)}
| not (x == 0) && helper == 0 => { ret == v;
    [x, 2] ** x :-> v ** x + 1 :-> xnxt ** sll_c(xnxt)
}
| not (x == 0) && not (helper == 0) => {
    [x, 2] ** x :-> v ** x + 1 :-> xnxt ** 
    [helper, 2] ** helper :-> vv ** helper + 1 :-> helpernxt **
    take_helper(xnxt, helpernxt, ret)
}
}
\end{lstlisting}

\subsection{replicate}

This has similar problems as \verb|take|: using a machine \verb|Int| for a count argument
instead of a unary representation.

\fsImpl

\begin{lstlisting}
%generate replicate [Int, Int] Sll

replicate : Int -> Int -> List;
replicate n i
  | n == 0 := Nil;
  | not (n == 0) := Cons i (replicate (n - 1) i);
\end{lstlisting}

\genSuSLik

\begin{lstlisting}
predicate replicate__rw_Sll__Int__Int(int __p_0, int __p_1
                                     ,loc __r_x) {
| (__p_0 == 0) => { __r_x == 0 ; emp }
| (not (__p_0 == 0)) => {
    replicate__rw_Sll__Int__Int((__p_0 - 1), __p_1, __p_x2) **
    __r_x :-> __p_1 ** (__r_x+1) :-> __p_x2 ** [__r_x,2] }
}
\end{lstlisting}

\refSuSLik

\begin{lstlisting}
predicate replicate_helper(int n, loc helper, loc ret){
| helper == 0 => {ret == 0 ; emp}
| not (helper == 0) => {
    [helper, 2] ** helper :-> vv ** helper + 1 :-> helpernxt **
    [ret, 2] ** ret :-> n ** ret + 1 :-> retnxt **
    replicate_helper(n, helpernxt, retnxt)
}
}
\end{lstlisting}

\subsection{foldMap}

This example maps a function over a list and then folds the result. This demonstrates nested function application.

\fsImpl

\begin{lstlisting}
%generate foldMap[Int, Sll[mutable]] Int

foldMap : Int -> List -> Int;
foldMap z (Nil) := z;
foldMap z (Cons x xs) :=
  instantiate [Int, Sll[mutable]] Int fold_List
     z
     (instantiate [Sll] Sll map
        (lower Sll (Cons x xs)));
\end{lstlisting}

\genSuSLik

\begin{lstlisting}
predicate foldMap__Int__Int__rw_Sll(int __p_0, loc __p_x1, int __r) {
| (__p_x1 == 0) => { __r == __p_0 ; emp }
| (not (__p_x1 == 0)) => {
     __p_x1 :-> x ** (__p_x1+1) :-> xs ** [__p_x1,2] **
     func fold_List__Int__Int__rw_Sll(__p_0, __p_x2, __r) **
     func map__rw_Sll__ro_Sll(__p_x2, __p_x2) **
     __p_x2 :-> x ** (__p_x2+1) :-> xs ** [__p_x2,2] **
     temploc __p_x2 }
}
\end{lstlisting}

%

\subsection{mapSum}
\label{sec:examples-mapSum}

This example takes a list of lists of \verb|Int|s and sums up each inner list.

\fsImpl

\begin{lstlisting}
%generate map_sum [ListOfListsLayout] Sll

map_sum : ListOfLists -> List;
map_sum (LNil) := Nil;
map_sum (LCons xs xss) :=
  Cons (instantiate [Sll] Int sum xs)
       (map_sum xss);
\end{lstlisting}

\genSuSLik

\begin{lstlisting}
predicate map_sum__rw_Sll__ro_ListOfListsLayout(loc __p_x0
                                               ,loc __r_x) {
| (__p_x0 == 0) => { __r_x == 0 ; emp }
| (not (__p_x0 == 0)) => {
     __r_x == __p_1
       ;
     __p_x0 :-> xs ** (__p_x0+1) :-> xss ** ro_Sll(xs) **
     [__p_x0,2] **
     func sum__Int__ro_Sll(xs, __p_1) **
     map_sum__rw_Sll__ro_ListOfListsLayout(xss, __p_x4) **
     (__r_x+1) :-> __p_x4 ** [__r_x,2] }
}
\end{lstlisting}

\refSuSLik

\begin{lstlisting}
predicate map_sum(loc x, loc y) {
|  x == 0        => { y == 0 ; emp }
|  not (x == 0)  => {
    [x, 2] ** x :-> h ** (x + 1) :-> xnxt ** sll_c(h) **
    [y, 2] ** func sum(x, y) ** (y + 1) :-> ynxt
    ** map_sum(xnxt, ynxt) }
}
\end{lstlisting}

\subsection{leftList}

This example takes a binary tree and gives back a list of items along the path going
down only the left subtrees. This demonstrates the interaction between two different
data types.

\fsImpl

\begin{lstlisting}
%generate leftList [TreeLayout] Sll

leftList : Tree -> List;
leftList (Leaf) := Nil;
leftList (Node a b c) := Cons a (leftList b);
\end{lstlisting}

\genSuSLik

\begin{lstlisting}
predicate leftList__rw_Sll__ro_TreeLayout(loc __p_x0, loc __r_x) {
| (__p_x0 == 0) => { __r_x == 0 ; emp }
| (not (__p_x0 == 0)) => {
   __p_x0 :-> a ** (__p_x0+1) :-> b ** (__p_x0+2) :-> c **
   [__p_x0,3] **
   leftList__rw_Sll__ro_TreeLayout(b, __p_x1) **
   __r_x :-> a ** (__r_x+1) :-> __p_x1 ** [__r_x,2] }
}
\end{lstlisting}

\refSuSLik

\begin{lstlisting}
predicate left_tree(loc x, loc y){
|  x == 0        => { emp }
|  not (x == 0)  => {
    [x, 3] ** x :=> v ** x + 1 :=> l ** x + 2 :=> r ** tree_c(r) **
    [y, 2] ** y :-> v ** y + 1 :-> ynxt ** left_tree(l, ynxt)
}
}
\end{lstlisting}

\subsection{sum}

This is effectively a fold specialized to the addition function and 0.

\fsImpl

\begin{lstlisting}
%generate sum [Sll] Int

sum : List -> Int;
sum (Nil) := 0;
sum (Cons head tail) := head + (sum tail);
\end{lstlisting}

\genSuSLik

\begin{lstlisting}
predicate sum__Int__ro_Sll(loc __p_x0, int __r) {
| (__p_x0 == 0) => { __r == 0 ; emp }
| (not (__p_x0 == 0)) => {
    __r == (head + __p_1) && __temp_0 == __p_1
      ;
    __p_x0 :-> head ** (__p_x0+1) :-> tail ** [__p_x0,2] **
    sum__Int__ro_Sll(tail, __temp_0) }
}
\end{lstlisting}

\refSuSLik

\begin{lstlisting}
predicate sum(loc l, int output){
| l == null => {output == 0; emp}
| l != null => {output == output1 + v; [l, 2] ** l :-> v **
    l + 1 :-> lnxt ** sum(lnxt, output1)}
}
\end{lstlisting}


\section{Translation Stages Example}
\label{sec:ImplEx}

We use the \verb|filterLt9| example to illustrate each stage. An additional syntactic construct is
used in this example to illustrate the internal changes to the AST as the translation goes through its
stages. The syntax \texttt{layout\{A\} \& e} is a \fnLang{}
expression \texttt{e} annotated with the SuSLik assertion \texttt{A}.

\begin{lstlisting}
%generate filterLt9 [Sll[readonly]] Sll

data List := Nil | Cons Int List;

Sll : List >-> layout[x];
Sll Nil := emp;
Sll (Cons head tail) := x :-> head, (x+1) :-> tail, Sll tail

filterLt9 : List -> List;
filterLt9 p Nil      := Nil;
filterLt9 p (Cons head tail)
  | head < 9       := filterLt9 p tail;
  | not (head < 9) := Cons head (filterLt9 p tail);
\end{lstlisting}

\begin{enumerate}
  \item \stage{Type checking and elaboration}.
  \begin{lstlisting}
filterLt9 (Sll[readonly ; x] Nil) := lower Sll[readonly ; r] Nil;
filterLt9 (Sll[readonly ; x] (Cons head tail))
| head < 9       :=
    instantiate [Sll[readonly ; tail]] Sll[r] filter tail;
| not (head < 9) :=
    lower Sll[readonly ; r]
      (Cons head
        (instantiate [Sll[readonly ; tail]] Sll[y]
          filterLt9 tail));
  \end{lstlisting}

  \item \stage{Unfold empty constructors}.
  \begin{lstlisting}
filterLt9 (Sll[readonly ; x] Nil) := lower Sll[readonly ; r] 0;
filterLt9 (Sll[readonly ; x] (Cons head tail))
| head < 9       :=
    instantiate [Sll[readonly ; tail]] Sll[r] filter tail;
| not (head < 9) :=
    lower Sll[readonly ; r]
      (Cons head
        (instantiate [Sll[readonly ; tail]] Sll[y]
          filterLt9 tail));
  \end{lstlisting}

  \item \stage{Unfold pattern matches using layouts}.
  \begin{lstlisting}
filterLt9 Nil      :=
    layout{ r == 0 ; emp }
      & 0;
filterLt9 (Cons head tail)
  | head < 9       :=
        layout{ x :=> head, (x+1) :=> tail }
          & instantiate [Sll[readonly ; tail]] Sll[r] filter tail;
  | not (head < 9) :=
        layout{ x :=> head, (x+1) :=> tail }
          & lower Sll[readonly ; r]
              (Cons head
                (instantiate [Sll[readonly ; tail]] Sll[y]
                  filterLt9 tail));
  \end{lstlisting}

  \item \stage{Insert copying predicate applications}. Not applicable.

  \item \stage{Translate \texttt{let}s}. Not applicable.

  \item \stage{Unfold constructor applications}.
\begin{lstlisting}
filterLt9 Nil      := layout{ r == 0 ; emp };
filterLt9 (Cons head tail)
  | head < 9       :=
      layout{ x :=> head, (x+1) :=> tail
        , filterLt9__Sll_Sll[tail | r] tail));
  | not (head < 9) :=
      layout{ x :=> head, (x+1) :=> tail, r :-> head, (r+1) :-> y,
        filterLt9__Sll_Sll[tail | y] tail));
\end{lstlisting}

  \item \stage{Generation}.
\begin{lstlisting}
inductive Sll(loc x) {
| x == 0 => { emp }
| not (x == 0) => { x :-> head ** (x+1) :-> tail ** Sll(tail) }
}

inductive filterLt9__Sll_Sll(loc x, loc r) {
| x == 0 => { r == 0 ; emp }
| not (x == 0) && head < 9 => {
    x :=> head ** (x+1) :=> tail ** filterLt9__Sll_Sll(tail, r)
  }
| not (x == 0) && not (head < 9) => {
    x :=> head ** (x+1) :=> tail ** r :-> head ** (r+1) :-> y **
    filterLt9__Sll_Sll(tail, y)
  }
}
\end{lstlisting}
\end{enumerate}


\section{\fnLang{} Grammar}
\label{sec:grammar}

\begin{grammar}
<int> ::= $\cdots$ | `-1' | `0' | `1' | $\cdots$

<natural> ::= `0' | `1' | $\cdots$

<bool> ::= `false' | `true'

<expr> ::= <int> \alt <bool> \alt <var> \alt <expr> <expr>
  \alt <expr> `+' <expr> \alt <expr> `-' <expr> \alt <expr> `\&\&' <expr> \alt <expr> `||' <expr> \alt `not' <expr>
  \alt `instantiate' <layout-list> <layout-name> <var> \{ <expr> \}
  \alt `lower' <moded-layout> <expr>

<generate-directive> ::= `\%generate' <var> <layout-list> <layout-name>

<layout-list> ::= `[' \{ <moded-layout> \} `]'

<mode> ::= `readonly' \alt `mutable'

<moded-layout> ::= <layout-name> `[' <mode> `]'

<base-type> ::= `Int' \alt `Bool'

<type> ::= <base-type> \alt <var> \alt <type> \verb|->| <type>

<fn-decl> ::= <var> `:' <type> `;'

<fn-def> ::= <fn-case> `;' <fn-def> | <fn-case> `;'

<fn-case> ::= <var> <pattern> <guarded-bodies>

<guarded-bodies> ::= <guarded-body> <guarded-bodies> | <guarded-body>

<guarded-body> ::= `|' <expr> `:=' <expr>

<layout-sig> ::= <layout-name> `:' <data-name> `\verb|>->|' `layout' `[' <vars> `]' `;'

<layout-def> ::= <layout-case> `;' <layout-def> | <layout-case> `;'

<layout-case> ::= <layout-name> <pattern> `:=' <layout-body>

<layout-body> ::= <heaplet> `,' <layout-body> | <heaplet>

<heaplet> ::= `emp' | <loc> `:->' <var> | <var> <var>

<loc> ::= <var> | `(' <var> + <natural> `)'

<pattern> ::= `(' <constructor-name> <var> `)'

<data-def> ::= `data' <data-name> `:=' <data-alts> `;'

<data-alts> ::= <data-alt> | <data-alt> `|' <data-alts>

<data-alt> ::= <constructor-name> \{ <type> \}
\end{grammar}


\section{Soundness Proof}
\label{sec:soundness-proof}

First, we introduce an operator that will make it more convenient to talk about the conjunction of two SSL propositions:
\[
  (p_1 ; s_1) \otimes (p_2 ; s_2) = (p_1 \land p_2 ; s_1 \sep s_2)
\]

\noindent
Next, we present a lemma regarding this operator that will be used in the soundness proof:
\begin{lemma}[$\otimes$ pairing]\label{thm:otimes-entail}
  If $(\sigma_1, h_1) \models p$ and $(\sigma_2, h_2) \models q$ and $\sigma_1 \subseteq \sigma_2$ and $h_1 \mathbin{\bot} h_2$,\\
  then $(\sigma_2, h_1 \circ h_2) \models p \otimes q$
\end{lemma}

\noindent
Now we proceed to soundness.

\soundnessThm*

\begin{proof}
  Proceeding by induction on $e$:
  \begin{itemize}
    \item $e = i$ for some $i \in \mathbb{Z}$.\\
      We have 
      \begin{itemize}
        \item $\Tsem{i}{V,r} = (v == i ; \emp)$
        \item $v \fresh V$
        \item $v = r$
        \item $\EndSigma = \StartSigma \cup \{ (r,i) \}$
        \item $\End{h'} = \Start{h} = \emptyheap$
      \end{itemize}
      Plugging these in, we find that we want to prove
      \[
        (\StartSigma \cup \{ (r,i) \}, \emptyheap) \models (r == i ; \emp)
      \]
      This immediately follows from the rules of separation logic.

    \item $e = b$ for some $b \in \mathbb{B}$. This case proceeds exactly as the previous case.

    \item $e = e_1 + e_2$.
      From the $\Tsem{\cdot}{}$ hypothesis, we have
      \begin{itemize}
        \item $(e_1, V_0) \tstep (p_1, s_2, V_1, v_1)$
        \item $(e_2, V_1) \tstep (p_2, s_2, V_2, v_2)$
        \item $v \fresh V_2$
        \item $v = r$
        \item $\Tsem{e_1 + e_2}{V, r} = (v == v_1 + v_2 \land p_1 \land p_2 ; s_1 \sep s_2)$
        \item $\Tsem{e_1}{V_0,v_1} = (p_1 ; s_1)$
        \item $\Tsem{e_2}{V_1,v_2} = (p_2 ; s_2)$
      \end{itemize}

      Plugging these in, we want to show
      \[
        \EndPair \models (v == v_1 + v_2 \land p_1 \land p_2 ; s_1 \sep s_2)
      \]

      Furthermore, we see that
      \[
        \Tsem{e}{V,r} = (v == v_1 + v_2 ; \emp) \otimes \Tsem{e_1}{V_0,v_1} \otimes \Tsem{e_2}{V_1,v_2}
      \]

      From the {\sc AM} relation hypothesis we have
      \begin{enumerate}
        \item $\EndSigma = \sigma_y \cup \{(r, z)\}$
        \item $z = x' + y'$
        \item $(x, \sigma, \mathcal{F}, h_1) \step (x', \sigma_x, h_1', \mathcal{F}, v_x)$
        \item $(y, \sigma_x, \mathcal{F}, h_2) \step (y', \sigma_y, h_2', \mathcal{F}, v_y)$
        \item $v_1 = v_x$
        \item $v_2 = v_y$
        \item $h = h_1 \circ h_2$
        \item $h' = h_1' \circ h_2'$
      \end{enumerate}

      From the first six items, we can derive
      \[
        (\EndSigma, \emptyheap) \models (v == v_1 + v_2 ; \emp)
      \]

      From the inductive hypotheses, we get
      \begin{itemize}
        \item $(\sigma_x, h_1') \models \Tsem{e_1}{\dom(\sigma_x), v_x}$
        \item $(\sigma_y, h_2') \models \Tsem{e_2}{\dom(\sigma_y), v_y}$
      \end{itemize}

      By Lemma~\ref{thm:otimes-entail}
      \[
        (\sigma_y, h_1' \circ h_2') \models \Tsem{e_1}{\dom(\sigma_x), v_x} \otimes \Tsem{e_2}{\dom(\sigma_y), v_y}
      \]

      and we know that $\sigma_x \subseteq \sigma_y$.

      By Lemma~\ref{thm:otimes-entail}, we conclude
      \[
        (\sigma_y \cup \{(r, z)\}, h') \models (v == v_1 + v_2 ; \emp) \otimes \Tsem{e_1}{\dom(\sigma_x),v_1} \otimes \Tsem{e_2}{\dom(\sigma_y),v_2}
      \]

    \item $e = v$ for some $v \in \Var$.
      Therefore the only two {\sc AM} rules that apply are {\sc AM-Base-Var} and {\sc AM-Loc-Var}. Both
      cases proceed in the same way.

      We know
      \begin{itemize}
        \item $h = h' = \emptyheap$
      \end{itemize}

      So we want to show
      \[
        (\EndSigma, \emptyheap) \models (\verb|true| ; \emp)
      \]
      This trivially holds.

    \item $e = \lowerS{A}{v}$ for some layout $A$ and $v \in \Var$.
      The applicable {\sc AM} rule is {\sc AM-Lower}.
      From the premises of {\sc AM-Lower}, we have
      \begin{itemize}
        \item $h' = \hat{h} \cdot H'$
      \end{itemize}
      where $H'$ comes from applying the layout $A$ to the constructor application expression
      obtained by reducing $e$. As a result, $H'$ exactly fits the specification of
      one of the branches of the $A$ layout.

      $v$ must be associated with some high-level value in $\mathcal{F}$ and $\hat{h}$
      is the part of the heap that is updated when this is reduced. Since all expressions
      are required to be well-typed, $\hat{h}$ must satisfy some collection of heaplets
      in the $A$ layout branch that is associated to $H'$.

      Note that $\Tsem{e}{V,r} = A(v)$

      From these facts, we can conclude
      \[
        \EndPair \models A(v)
      \]

    \item $e = \lowerS{A}{C\; e_1 \cdots e_n}$. This case is similar to the previous case,
      except that we do not need to lookup a variable in the store before dealing with the constructor
      application.

    \item $e = \instantiateS{A,B}{f}(v)$ for some $v \in \Var$.
        The rule {\sc AM-Instantiate} applies here.

        The {\sc S-Inst-Var} rule requires that we satisfy the inductive predicate
        associated to $f$, that is $\mathcal{I}_{A,B}(f)(v)$. Note that the SSL propositions in the definition of
        such an inductive predicate, given by {\sc FnDef}, will always not only explicitly
        describe the memory used in its result (given by the layout $B$) but also the memory used
        in its argument (given by the layout $A$).

        Now, consider that the {\sc AM-Instantiate} rule uses the layout $A$ to the function argument to
        update the heap and uses the $B$ layout to produce the function's result on the heap. This will match the
        inductive predicate associated to $f$ instantiated at the layouts $A$ and $B$, as required.

      %
      %
      %


      \item $e = \instantiateS{A,B}{f}(C\; e_1 \cdots e_n)$.
        This is much like the previous case. The only difference is that the {\sc S-Inst-Constr} rule
        unfolds the inductive predicate $\mathcal{I}_{A,B}(f)(C\; e_1 \cdots e_n)$. The resulting SSL proposition
        will still be satisfied by the model $\EndPair$ given by {\sc AM-Instantiate}.

      \item $e = \instantiateS{B,C}{f}(\instantiateS{A,B}{g}(e_0))$.
        This case combines to instantiates together. The main condition we need to check here is that the
        two instantiate applications use disjoint parts of the heap in $\EndPair$ in {\sc AM-Instantiate}.

        This can be seen to be true by the fact that the resulting heap is built up out of the subexpressions
        using $\circ$, ensuring that the parts of the heap are disjoint.
  \end{itemize}
\end{proof}


\end{document}